\documentclass[onecolumn,11pt,aps,nofootinbib,superscriptaddress,tightenlines]{revtex4}

\usepackage{graphics}
\usepackage{dsfont}
\usepackage{amsmath}
\usepackage{amssymb}
\usepackage{amsthm}
\usepackage{graphicx}
\usepackage{mathrsfs}
\usepackage{units}
\usepackage[usenames,dvipsnames]{color}
\usepackage{tikz}
\usetikzlibrary{positioning,decorations.pathreplacing,shapes,arrows}
\newcommand{\abs}[1]{\ensuremath{|#1|}}

\newcommand{\Abs}[1]{\ensuremath{\left|#1\right|}}

\newcommand{\norm}[2]{\ensuremath{|\!|#1|\!|_{#2}}}

\newcommand{\Norm}[2]{\ensuremath{\left|\!\left|#1\right|\!\right|_{#2}}}

\newcommand{\tr}{\textnormal{tr}}
\newcommand{\trace}[1]{\ensuremath{\tr (#1)}}
\newcommand{\Trace}[1]{\ensuremath{\tr \left( #1 \right)}}

\newcommand{\ptr}[1]{\textnormal{tr}_{#1}\,}
\newcommand{\ptrace}[2]{\ensuremath{\ptr{#1} (#2)}}

\newcommand{\idx}[2]{{#1}_{#2}}

\newcommand{\ket}[1]{| #1 \rangle}
\newcommand{\keti}[2]{| #1 \rangle_{#2}}

\newcommand{\braket}[2]{\langle #1 | #2 \rangle}

\newcommand{\bracket}[3]{\langle #1 | #2 | #3 \rangle}

\newcommand{\proj}[2]{| #1 \rangle\!\langle #2 |}
\newcommand{\proji}[3]{| #1 \rangle\!\langle #2 |_{#3}}

\newcommand{\id}{\ensuremath{\mathds{1}}}
\newcommand{\idi}[1]{\ensuremath{\idx{\mathds{1}}{#1}}}

\newcommand{\cB}{\mathcal{B}}

\newcommand{\cG}{\mathcal{G}}
\newcommand{\cH}{\mathcal{H}}

\newcommand{\cO}{\mathcal{O}}

\newcommand{\cT}{\mathcal{T}}

\newcommand{\upp}[2]{\ensuremath{#1}^{(#2)}}
\theoremstyle{plain}
\newtheorem{lemma}{Lemma}
\newtheorem{theorem}[lemma]{Theorem}
\newtheorem{corollary}[lemma]{Corollary}

\theoremstyle{definition}

\begin{document}
\title{Perturbation Theory for Parent Hamiltonians of Matrix Product States}
\author{Oleg Szehr}
\email{oleg.szehr@posteo.de}
\affiliation{Zentrum Mathematik, Technische Universit\"{a}t M\"{u}nchen, 85748 Garching, Germany}
\author{Michael M. Wolf}
\email{wolf@ma.tum.de}
\affiliation{Zentrum Mathematik, Technische Universit\"{a}t M\"{u}nchen, 85748 Garching, Germany}
\date{February 24, 2014}
\begin{abstract}
This article investigates the stability of the ground state subspace of a canonical parent Hamiltonian of a Matrix product state against local perturbations. We prove that the spectral gap of such a Hamiltonian remains stable under weak local perturbations even in the thermodynamic limit, where the entire perturbation might not be bounded. Our discussion is based on preceding work by D.A.\:Yarotsky that develops a perturbation theory for relatively bounded quantum perturbations of classical Hamiltonians. We exploit a renormalization procedure, which on large scale transforms the parent Hamiltonian of a Matrix product state into a classical Hamiltonian plus some perturbation. We can thus extend D.A.\:Yarotsky's results to provide a perturbation theory for parent Hamiltonians of Matrix product states and recover some of the findings of the independent contributions \cite{Stable,Rob}.
%. A similar result was achieved independently and with different methods in the contributions \cite{Stable,Rob}.

%Many important models of quantum spin lattices such as the AKLT model are given by the parent Hamiltonians of respective Matrix product states. The corresponding stability results can be seen as special cases of the more general framework described in this article.
\end{abstract}

\maketitle

\tableofcontents
\section{Introduction}
The purpose of this article is to investigate the low energy sector of certain models of many-body quantum systems with local interaction. We are interested in the stability of quantum phases when small perturbations act on the system. In particular, we aim at understanding the conditions under which certain physical properties of the ground state change smoothly when an interaction is added to the model Hamiltonian. In this article we study stability of the ground state in the so-called parent Hamiltonian model~\cite{MPSwerner}.

Matrix product states (MPS) have been an extremely useful tool in the study of the ground state physics of many-body quantum systems.
% systems arising from local quantum Hamiltonians.
With their local structure MPS provide an efficient description of states arising from local interactions and constitute a natural framework for the analysis of local gapped Hamiltonians in 1D. In fact, the matrix product state representation lies at the heart of the very successful density matrix renormalization group method \cite{DMRG1,DMRG2}. To any MPS a local frustration-free and gapped Hamiltonian having this MPS as a unique ground state can be associated.
A canonical choice of such Hamiltonians was introduced in \cite{MPSwerner} and is referred to as \emph{parent Hamiltonian} of the MPS. On the one hand the local structure of the MPS endows the canonical parent Hamiltonian with the structure necessary for a rigorous analysis.
% of the ground state properties. 
On the other hand canonical parent Hamiltonians constitute a wide class of local Hamiltonians and include many important special cases such as the AKLT-Hamiltonian~\citep{AKLT}.

We are interested in how the parent Hamiltonian model behaves under small perturbations, as this allows one to use the idealization to predict the behaviour of actual physical systems.
It seems generally expected that if a ground state of a quantum many-body system is in a non-critical regime characterized by the presence of a local spectral gap and exponential decay of correlations, then the system remains in this phase under sufficiently weak perturbations.
We prove that for translationally invariant parent Hamiltonians of generic MPS this is indeed the case i.e.~we show that the spectral gap of such a Hamiltonian is stable under arbitrary local perturbations even in the thermodynamic limit. This result itself is not new. It was shown in \cite{Stable} that local Hamiltonians that satisfy the \emph{Local Topological Quantum Order} (LTQO) condition and that are \emph{locally gapped} are stable under local perturbations. It was also claimed in \cite{Stable} and shown in \cite{Rob} that parent Hamiltonians of MPS have LTQO. (However, in spin systems of higher spatial dimension the presence of LTQO is hard to verify.) The fact that parent Hamiltonians are locally gapped was already known from \cite{Nacht1}. Hence, the stability of the spectral gap against sufficiently weak perturbations follows.

The contribution at hand contains a new proof of this result. Our derivation is based on the observation that with increasing system scale a matrix product state \lq\lq{}looks more and more classical\rq\rq{} \cite{RG}. %Presumably the corresponding parent Hamiltonians should also evolve to a classical Hamiltonian. 
We exploit a renormalization group flow on parent Hamiltonians to prove that on sufficiently large scale a (generic) parent Hamiltonian can be seen as a perturbation of a classical system. Hence, any sufficiently small quantum perturbation of a parent Hamiltonian is equivalent to a relatively bounded perturbation of a classical model. We then draw on the theory for ground states in quantum perturbations of classical lattice systems by D.A. Yarotsky \cite{YAR} to conclude our proof. %We prove that on sufficiently large scale the parent Hamiltonian can be seen as a perturbation of a classical model and we apply stability results for classical systems under quantum perturbations to prove
The results presented in this article were achieved independently of the contributions \cite{Stable,Rob}, before the publication of the latter.

%We conclude this introduction with a quotation by D.A. Yarotsky about his proof of stability of the AKLT model~\cite{YAR}:
%
%\begin{quote}
%\lq\lq{}The main point of the proof is that on a large length scale the AKLT model is a relatively bounded perturbation of a classical model [...]. Though in this paper we restrict our attention to the AKLT model only, this property is definitely more general; one can expect some form of it to be generic to non-critical gapped spin systems.\rq\rq{}
%\end{quote}
%
%
%
%
%
%
%
%
%
%
%
\section{Preliminaries}
\label{prel}
As mentioned in the introduction, this article investigates how the ground state subspace of an MPS parent Hamiltonian behaves under small perturbations. This section reviews the required definitions and basic results.
\subsection{Notation}
\label{prel:not}

We model quantum spin chains as connected subsets $\Lambda\subset \mathbb{Z}$, where each site $x\in\Lambda$ is equipped with a $d$-dimensional, complex Hilbert space $\cH_x$. We denote by $\Lambda_k$ with $k\in\mathbb{N}$ connected subsets of $\Lambda$ and by $\abs{\Lambda_k}$ the number of sites in $\Lambda_k$. The total Hilbert space associated to a finite set $\Lambda\subset\mathbb{Z}$ will be denoted by $\cH_\Lambda=\bigotimes_{x\in\Lambda}\cH_x$.
The interactions on the spin chain are given by a translationally invariant (TI) Hamiltonian with some fixed interaction range $\Lambda_0$.
Such Hamiltonians can formally be written as
\begin{align*}
H_\Lambda=\sum_{x\in\Lambda} h_x,
\end{align*}
where $h_x$ is a positive semi-definite operator acting (non-trivially) on $\cH_{\Lambda_0+x}$ and $\Lambda_0+x$ is a translate of $\Lambda_0$ by $x$.  We will assume that $H_\Lambda$ has a non-degenerate ground state $\ket{\Omega}_\Lambda$ and that $H_\Lambda$ has a \emph{spectral gap} $\gamma>0$ above the ground state energy
$$H_\Lambda\big|_{\cH_\Lambda\ominus\ket{\Omega}_\Lambda}\geq\gamma\:\id.$$

Moreover, the Hamiltonians considered in this article will be \emph{frustration free}, that is each interaction term $h_x$ minimizes the global ground state energy: for all $x$ we have $h_x\ket{\Omega}_\Lambda=0$. We analyse how the spectral gap behaves if the Hamiltonian is perturbed with local interactions. Formally, we add a perturbation
$$\Phi_\Lambda=\sum_{x\in\Lambda}\phi_x,$$
where each of the terms $\phi_x$ acts locally on a finite subset of $\Lambda$.
Often, we will find it convenient to identify the first and last site of $\Lambda$ to impose \emph{periodic boundary conditions} (PBC) on the system.

To distinguish particular Hilbert subspaces of $\cH_\Lambda$ we will add Latin subscripts, for example $\cH_A$ and $\cH_{B}$. For any operator $X$ acting on a finite subset of the chain we denote by $\norm{X}{p}$ the Schatten $p$-norm of $X$. If $X$ acts on an infinite subsets we will only employ the $\norm{\cdot}{\infty}$-norm, which coincides with the usual operator norm.

As mentioned before we will consider a renormalization group flow that transforms the MPS parent Hamiltonian into a classical Hamiltonian. This flow will be modeled using a consecutive application of a linear map $\cT$ acting on matrices $X$. More precisely, we define the map $\cT$ by $\cT(X):=\sum_i A_iXA_i^\dagger$, where the summation goes over a set of so-called Kraus operators $\{A_i\}_i$. Maps with this structure are completely positive (CP). For each such map the dual map $\cT^*$ is defined by $\cT^*(X):=\sum_i A_i^\dagger XA_i$. $\cT^*$ is simply the adjoint of $\cT$ with respect to the Hilbert-Schmidt inner product $\braket{X}{Y}=\trace{X^{\dagger} Y}$. $\cT$ is called unital (CPU) iff it preserves the identity operator $\cT(\id)=\id$ and $\cT$ is called trace-preserving (CPTP) iff $\cT^*(\id)=\id$. 
%We will often consider specific matrix representations of $\cT$ and $\cT^*$ given by $T:=\sum_iA_i\otimes \bar{A}_i$ and $T^*:=\sum_iA_i^\dagger\otimes \bar{A}_i^\dagger$, respectively. 
%
%
%
%
%
%
\subsection{Matrix Product States}\label{prel:MPS}
We consider a finite subset $\Lambda\subset\mathbb{Z}$ consisting of $N$ sites, whose Hilbert spaces are each of dimension $d$. Every pure state of the spin system of $\Lambda$ can be written as 
\begin{align*}
\ket{\Psi}=\sum_{i_1,...,i_N=1}^d\trace{A_{i_1}^{[1]}\cdot A_{i_2}^{[2]}\cdot...\cdot A_{i_N}^{[N]}}\ket{i_1...i_N}
\end{align*}
with site dependent $D_k\times D_{k+1}$ matrices $A_{i_k}^{[k]}$ \cite{Vidal03,MPS}. States of this structure are called Matrix product states. In the case of periodic boundary conditions and translational invariance of the MPS it is possible to show \cite{MPS} that the matrices can be chosen in a site-independent way, i.\:e.\:
\begin{align*}
\ket{\Psi}=\sum_{i_1,...,i_N=1}^d\trace{A_{i_1}\cdot A_{i_2}\cdot...\cdot A_{i_N}}\ket{i_1...i_N}
\end{align*}
with $D\times D$ matrices $\{A_i\}_{i=1,...d}$. In our consecutive discussion a special class of MPS will be of particular importance. This class is characterized by the following generic condition.
\\
\\
\textbf{Condition (G1):}\\
There is a finite number $L_0$ such that for all $L\geq L_0$ the list of matrices
\begin{align*}
\{A_{i_1}\cdot...\cdot A_{i_{L}}\}_{i_j\in\{1...d\}}
\end{align*}
spans the entire algebra of $D\times D$ matrices.\\
\\
Condition (G1) is generic in the sense that $d$  matrices chosen randomly according to some reasonable measure comply with this condition with probability one. It is not hard to see that (G1) holds iff the map
$$\Gamma_L: X\mapsto \sum_{i_1,...,i_L=1}^d\trace{XA_{i_1} A_{i_2}... A_{i_L}}\ket{i_1...i_L}$$
is injective for $L\geq L_0$. The correspondence between sets $\{A_i\}_{i=1,...d}$ and MPS is not bijective; for example the set $\{XA_iX^{-1}\}_{i=1,...d}$ with invertible $X$ belongs to the same state. It is shown in \cite{MPS},~Chapter~3 that the matrices of any MPS satisfying (G1) can be chosen to constitute a CPU map $\cT$. More precisely, we can choose $\{A_i\}_{i=1,...d}$ such that the map $\cT(X)=\sum_iA_iXA_i^\dagger$ satisfies $\cT(\id)=\id$ and $\cT^*(\Xi)=\Xi$ for some diagonal and strictly positive matrix $\Xi$. In addition, $\id$ is the only fixed point of $\cT$. For a more detailed discussion of MPS we refer to \cite{MPS}.
\subsection{Canonical Parent Hamiltonians}\label{prel:parent}
We consider a TI state $\ket{\Psi}=\sum_{i_1...i_N}\trace{A_{i_1}\cdot...\cdot A_{i_N}}\ket{i_1...i_N}$ of a spin system with PBC on a chain $\Lambda$ of $N$ sites. For fixed $L\in\mathbb{N}$ we define $\cG_L\subset(\mathbb{C}^d)^{\otimes L}$ to be the subspace spanned by the vectors $\ket{\Psi(X)}=\sum_{i_1...i_L}\trace{XA_{i_1}\cdot...\cdot A_{i_L}}\ket{i_1...i_L}$, where $X$ are complex $D\times D$ matrices. Note that if condition (G1) holds for the matrices $A_i$ then for $L\geq L_0$ the space spanned by $\ket{\Psi(X)}$ has dimension $D^2$. We write $h_{\cG_L}$ for the projector onto the orthogonal complement of $\cG_L$ in $(\mathbb{C}^d)^{\otimes L}$. The \textit{canonical parent Hamiltonian} for $\ket{\Psi}$ (and fixed $L$) is defined as the formal expression $H_\Lambda=\sum_i^N\tau^i(h_{\cG_L})$ where $\tau$ denotes the translation operation by one site \cite{MPSwerner,MPS}. For a parent Hamiltonian with nearest neighbour interaction ($L$=2) we will write $H_\Lambda=\sum_k h_{k,k+1}$ to emphasize this fact. It is clear from the definition that $H_\Lambda\ket{\Psi}=0$ and that $H_\Lambda$ is frustration free. Moreover, as a result of condition (G1) $\ket{\Psi}$ is the unique ground state of $H_\Lambda$ if $L> L_0$ and $N\geq2L_0$, \cite[Theorem 10]{MPS}.
More generally, under (G1) $H_\Lambda$ can be shown to have a spectral gap $\gamma>0$ above the ground state energy \cite{MPSwerner,MPS} even in the limit of an \emph{infinite} chain. Let $\Lambda_1\subset\Lambda$ and let $G_{\Lambda_1}$ denote the projector onto the kernel of $H_{\Lambda_1}=\sum_{i:\{i+1,...,i+L\}\subset\Lambda_1}\tau^i(h_{\cG_L})$. The \emph{local gap} is defined to be the largest number $\gamma_{\Lambda_1}$ such that
$$H_{\Lambda_1}\geq\gamma_{\Lambda_1}\left(\id-G_{\Lambda_1}\right).$$
The local gap does not depend on $\Lambda$ but only on the number of sites in $\Lambda_1$. The "Local-Gap condition" of \cite{Stable} refers to the property  of a general frustration-free Hamiltonian that the local gap decays at most polynomially in the number of lattice sites. It is one core assumption for the stability proof for frustration-free Hamiltonians (the other one being LTQO). In \cite{Nacht1,spitz} a constant lower bound on the local gap of one-dimensional, frustration-free Hamiltonians is derived. In particular, this implies that parent Hamiltonians satisfy the Local-Gap condition and we will naturally encounter this fact in our derivation. A more detailed discussions of parent Hamiltonians for MPS can be found in \cite{MPS}.

\subsection{Stability of the spectral gap under quantum Perturbations of classical Hamiltonians}
In this section we recall a fundamental result by D.A.\:Yarotsky \cite{YAR} that asserts the stability of the spectral gap of a classical Hamiltonian under certain local perturbations. The effect of small quantum perturbations to classical Hamiltonians was discussed for example in \cite{YAR2,YAR3,ALA1,KT3}. In \cite{KTHal,YAR} this was extended to perturbations that need not necessarily be small but are required to consist of a small bounded part and a term that is bounded relatively to the unperturbed Hamiltonian. In the following we describe rigorously this perturbation theory.

We start with a chain $\Lambda\subset\mathbb{Z}$ with PBC and we consider a TI frustration-free Hamiltonian $H_\Lambda=\sum_{x\in\Lambda}h_x$. We will call $H_\Lambda$ \emph{classical} if in each space $\cH_x$ there is a preferred vector $\ket{\Omega}_x$ and an orthogonal basis containing that vector such that the product basis in $\cH_{\Lambda_0+x}$ diagonalizes $h_x$. %Furthermore we assume that $\ket{\Omega}_{\Lambda_0+x}=\bigotimes_{y\in\Lambda_0+x}
%\ket{\Omega}_y$ is a non-degenerate gapped ground state of $h_x$.
%
Furthermore we assume that $H_\Lambda$ has non-degenerate ground state $\ket{\Omega}_\Lambda=\bigotimes_{x\in\Lambda}
\ket{\Omega}_x$ and strictly positive spectral gap above $\ket{\Omega}_\Lambda$. We consider perturbations $\Phi_\Lambda=\sum_{x\in\Lambda}\phi_x$ whose local terms act on finite subchains and that can be split into a purely bounded part $\phi_x^{(b)}$ and a relatively bounded part $\phi_x^{(r)}$ as
\begin{align}
\phi_x=\phi_x^{(r)}+\phi_x^{(b)}\label{seps}.
\end{align}
The bounded part is characterized by 
\begin{align}
\norm{\phi_x^{(b)}}{\infty}\leq\beta.\label{boundedpart}
\end{align}
For the relatively bounded part we suppose that for any $\ket{\psi}$ and any $I\subset\Lambda$
\begin{align}
\Abs{\sum_{x\in I}\bracket{\psi}{\phi_x^{(r)}}{\psi}}\leq\alpha\bracket{\psi}{H_\Lambda}{\psi}.\label{relboundedpart}
\end{align}

\begin{theorem}[{\cite[Theorem 2]{YAR}}]\label{lem:yar} Let ${H}_{\Lambda}=\sum_x h_x$ be a classical Hamiltonian on a chain $\Lambda$ with PBC and non-degenerate gapped ground state $\ket{\Omega}_\Lambda$. Consider the perturbed Hamiltonian $\widetilde{H}_\Lambda=H_{\Lambda}+\Phi$, where $\Phi=\sum_{x}\phi_x$ is a perturbation that satisfies \eqref{seps}-\eqref{relboundedpart}. 
For any $\kappa>1$ there is $\delta(\kappa)>0$ such that for any $\alpha\in(0,1)$ and $\beta=\delta(1-\alpha)^{2\kappa}$ the following conclusions hold:
\begin{enumerate}
\item $\widetilde{H}_\Lambda$ has a non-degenerate gapped ground state $\ket{\widetilde{\Omega}}_\Lambda$:
\begin{align*}
\widetilde{H}_\Lambda\ket{\widetilde{\Omega}}_\Lambda=
\widetilde{E}_\Lambda\ket{\widetilde{\Omega}}_\Lambda
\end{align*}
and for some $\gamma>0$ that does not depend on $\Lambda$
$$\widetilde{H}_\Lambda|_{\cH_\Lambda\ominus\ket{\widetilde{\Omega}}_\Lambda}\geq
(\widetilde{E}_\Lambda\:+\:\gamma)\:\id.$$
\item There exists a thermodynamic weak$^*$-limit of the ground states $\ket{\widetilde{\Omega}}_\Lambda$: For $\Lambda\rightarrow\mathbb{Z}$ one has that
\begin{align*}
\braket{A\widetilde{\Omega}_\Lambda}{\widetilde{\Omega}_\Lambda}\rightarrow\omega(A),\qquad A\in\bigcup_{\abs{\Lambda}<\infty}\cB(\cH_\Lambda),
\end{align*}
where $\cB(\cH_\Lambda)$ denotes the bounded operators on $\cH_\Lambda$.
\item There is an exponential decay of correlations in the infinite volume ground state $\omega$: for $\Lambda_1,\Lambda_2\subset\Lambda$ and some positive $c$ and $\epsilon<1$ it holds that
\begin{align*}
\abs{\omega(A_1A_2)-\omega(A_1)\omega(A_2)}\leq c^{\abs{\Lambda_1}+\abs{\Lambda_2}}\epsilon^{\textnormal{dist}(\Lambda_1,\Lambda_2)}\norm{A_1}{\infty}\norm{A_2}{\infty},\quad A_i\in\cB(\cH_{\Lambda_i}).
\end{align*}
\item If within the allowed range of perturbations the term $\phi_x$ depends analytically on some parameters, then the ground state $\omega$ is also weakly$^*$-analytic in these parameters.
\end{enumerate}
\end{theorem}

Theorem~\ref{lem:yar} establishes that the spectral gap of a classical Hamiltonian is stable under perturbations that comply with the above assumptions. We will use this result to prove that parent Hamiltonians of MPS have a spectral gap that is stable under sufficiently weak bounded perturbations. To achieve this we will view the MPS parent Hamiltonian as a perturbation of a classical Hamiltonian, which is within a parameter range where Theorem~\ref{lem:yar} applies. The bounded part of this perturbation will decay faster under scaling of the system size than $\delta(1-\alpha)^{2\kappa}$. For sufficiently large systems this implies that under a small bounded perturbation $\phi_x'$ the parent Hamiltonian remains a perturbation of a classical Hamiltonian such that Theorem~\ref{lem:yar} applies. This provides us with the desired perturbation result.

\section{Stability of the spectral gap of a canonical parent Hamiltonian}
In this section we state our main theorem. We consider a MPS that satisfies the generic condition (G1) and prove that the spectral gap of the corresponding parent Hamiltonian is stable under sufficiently weak perturbations. In the following corollary we extend this result and show that our discussion includes D.A.~Yarotsky's perturbation theory for the AKLT model \cite{YAR} as an important special case.
\begin{theorem}\label{thm1}
Let $\ket{\Psi}$ be a TI MPS on a finite ring $\Lambda$ with PBC and suppose that for the matrices of $\ket{\Psi}$ condition (G1) holds. Suppose $N\geq2L_0$ and choose $L> L_0$ and let $H_\Lambda=\sum_i\tau^i(h_{G_L})$ be the canonical parent Hamiltonian for $\ket{\Psi}$. Furthermore let $\Phi_\Lambda=\sum_k\phi_k$ be any finite range interaction with $\norm{\phi_k}{\infty}\leq\beta$ for a sufficiently small $\beta$ depending on the range of $\Phi$. Then all conclusions of Theorem~\ref{lem:yar} hold for the perturbed parent Hamiltonian $\widetilde{H}_\Lambda=H_\Lambda+\Phi_\Lambda$.
\end{theorem}
Note that the above does not apply to important special cases as the AKLT model. There one considers a Hamiltonian with local nearest neighbour interaction but the matrices at each site do not span the whole algebra. The following simple corollary is to remedy this issue.
\begin{corollary}\label{non-generic}
Let $H_\Lambda=\sum_i\tau^i(h_{G_L})$ be a canonical parent Hamiltonian such that Theorem~\ref{thm1} applies.
%and 
Consider a Hamiltonian $\hat{H}_\Lambda=\sum_i h_{i,i+1}$ and suppose that there are positive constants $c_1$ and $c_2$ such that
\begin{align*}
c_1\:h_{G_L}\leq\sum_{j=1}^{L-1}h_{j,j+1}\leq c_2\:h_{G_L}.
\end{align*}
Then all conclusions of Theorem~\ref{lem:yar} also hold for $\hat{H}$.
\end{corollary}
The ground states of the AKLT model are MPS with $\{A_i\}=\{\sigma^z,\sqrt{2}\sigma^+,-\sqrt{2}\sigma^-\}$ \cite{AKLT,MPS}, where the $\sigma$'s are the Pauli matrices. If we choose $\hat{H}$ to be he AKLT Hamiltonian Corollary~\ref{non-generic} applies with $L=3$ and implies the stability of the spectral gap of the AKLT model.
\section{Proof of Stability}

We start this section with an outline of the proof of Theorem~\ref{thm1}. In Section~\ref{teclem} we prove some lemmas from the theory of quantum channels and MPS. The following Subsection~\ref{coreproof} contains a proof of Theorem~\ref{thm1} under the stronger assumption that the matrices $\{A_i\}_{i=1,...,d}$ at each site of the chain span the whole algebra of $D\times D$ matrices. However, this assumption is not necessary and in Section~\ref{finishproof} we extend the previous discussion to prove stability under (G1).
\subsection{Outline of the proof}
For the readers convenience, before we proceed with the derivation of Theorem~\ref{thm1}, we start with an exposition of core observations that will provide us with the proof.
\begin{enumerate}
\item We are given a MPS parent Hamiltonian $H_\Lambda$. We divide $\Lambda$ into subchains $\Lambda_k$ and we consider local sub-Hamiltonians $H_{\Lambda_k\cup\Lambda_{k+1}}$ of $H_\Lambda$ acting on $\Lambda_k\cup\Lambda_{k+1}$. We analyze the behavior of the ground state subspace of $H_{\Lambda_k\cup\Lambda_{k+1}}$ under scaling of $\Lambda_k$. To this end we introduce density matrices $\rho_{\Lambda_k\cup\Lambda_{k+1}}$ whose image subspace is exactly the kernel of $H_{\Lambda_k\cup\Lambda_{k+1}}$.
\item Using a renormalization group flow we construct local unitaries $W_{\Lambda_k}$ such that on sufficiently large scale the image of $W_{\Lambda_k}\otimes W_{\Lambda_{k+1}}\rho_{\Lambda_k\cup\Lambda_{k+1}}
W_{\Lambda_k}^\dagger\otimes W_{\Lambda_{k+1}}^\dagger$ has particularly simple structure. It turns out that in the asymptotic limit of large system size this image corresponds to the ground state subspace of a classical Hamiltonian. %In fact we construct a classical Hamiltonian with the corresponding ground state subspace.
\item \sloppypar{We use convergence estimates from the theory of quantum Markov chains to show that the projectors $G_{\Lambda_k\cup\Lambda_{k+1}}$ onto the kernel of $H_{\Lambda_k\cup\Lambda_{k+1}}$ and $G_{\Lambda_k\cup\Lambda_{k+1}}^{(\infty)}$ onto the kernel of the asymptotic classical Hamiltonian can be made exponentially close. We prove that
\mbox{{$\norm{W_{\Lambda_k}\otimes W_{\Lambda_{k+1}}G_{\Lambda_k\cup\Lambda_{k+1}} W_{\Lambda_k}^\dagger\otimes W_{\Lambda_{k+1}}^\dagger-G_{\Lambda_k\cup\Lambda_{k+1}}^{(\infty)}}{\infty}\leq \cO(\abs{\lambda_2}^{L/2})$}}.}
\item We provide an explicit perturbation consisting of a bounded part $\sum_k\phi^{(b)}_k$ and a relatively bounded part $\sum_k\phi^{(r)}_k$ that transform the classical Hamiltonian into $\bigotimes_kW_{\Lambda_k}H_{\Lambda}
\bigotimes_{k}W_{\Lambda_k}^{\dagger}$.  Using the estimate from \emph{3.}~we show that these perturbations are in accordance with the conditions of Theorem~\ref{thm1}. When adding a sufficiently small bounded perturbation to $\sum_k\phi^{(b)}_k$ the total perturbation remains in the range where Theorem~\ref{thm1} applies. Hence, the ground state subspace of $H_{\Lambda}$ is stable.
\end{enumerate}
\subsection{Some Lemmas}\label{teclem}
%
%We start this section with a simple lemma about the norm of a linear map  acting on matrices and its matrix representation. We relate the norm of complete boundedness (CB-norm, \cite{Paulsen,Watrous}) of $\cE$ to the Schatten $\infty$-norm of the matrix representation $E$ of $\cE$ (see Section~\ref{prel:not}).
%
%
%\begin{lemma}\label{lem:one}Let $\cE: \textnormal{Mat}(D\times D,\mathbb{C})\mapsto \textnormal{Mat}(D\times D, \mathbb{C})$ be a linear map operating on complex $D\times D$ matrices. Let $E$ be its matrix representation defined by $E\ket{ij}:=\cE(\proj{i}{j})$. Then
%
%\begin{align}
%\norm{\cE}{CB}\leq D\norm{E}{\infty}
%\end{align}
%\end{lemma}
%
%
%
%\begin{proof} We have that
%
%
%\begin{align*}
%\norm{\cE}{CB}=\sup_{X\neq0}\frac{\norm{(\cE\otimes\id)(X)}{\infty}}{\norm{X}{\infty}}\leq D\sup_{X\neq0}\frac{\norm{(\cE\otimes\id)(X)}{2}}{\norm{X}{2}}=D\sup_{X\neq0}\frac{\norm{\cE(X)}{2}}{\norm{X}{2}}.
%\end{align*}
%
%
%The last equality is proven in \cite{Watrous}, Result 3. (This is just the fact that the $2\rightarrow2$ norm is stable under tensoring with the identity.) But $\sup_{X\neq0}\frac{\norm{\cE(X)}{2}}{\norm{X}{2}}$ simply equals the largest singular value of $T$, i.e. we have $\sup_{X\neq0}\frac{\norm{\cE(X)}{2}}{\norm{X}{2}}=\Norm{E}{\infty}$.
%\end{proof}
%
%We did not have to use the fact that $\cT$ is CP for the proof of the lemma and it clearly holds for general linear maps $\cT$ (with an appropriate definition of $T$).
%
%
%

We already mentioned (Section~\ref{prel:MPS}) that to any TI MPS we can associate a certain CPU map $\cT$. To better keep track of the kernel of the canonical parent Hamiltonian it will be useful to introduce the operator $\idx{\rho}{EE'}=\frac{1}{D}\sum_{i_1, i_2, j_1, j_2}\trace{A_{i_1}A_{i_2}A_{j_2}^\dagger A_{j_1}^\dagger}\proji{i_1}{j_1}{E}\otimes\proji{i_2}{j_2}{E'}$, which is defined via the Kraus operators of $\cT$. The subscripts $E$ and $E'$ have no physical significance but are introduced to more conveniently distinguish the systems involved. The following lemma shows that if two CPU maps $\cT$ and $\widetilde{\cT}$ are close, then the corresponding operators $\idx{\rho}{EE'}$
and $\idx{\widetilde{\rho}}{EE'}$ can be made close using a local unitary transformation.
\begin{lemma}\label{lem:two}
Let $\cT(X)=\sum_{i=1}^dA_iX{A_i}^{\dagger}$ and $\widetilde{\cT}(X)=\sum_{i=1}^d\widetilde{A}_iX{\widetilde{A}_i^{\dagger}}$ be CPU maps. Consider the operators $\idx{\rho}{EE'}:=\frac{1}{D}\sum_{i_1, i_2, j_1, j_2=1}^d\trace{A_{i_1}A_{i_2}A_{j_2}^{\dagger}
A_{j_1}^{\dagger}}\proji{i_1}{j_1}{E}\otimes\proji{i_2}{j_2}{E'}$ and $\idx{\widetilde{\rho}}{EE'}:=\frac{1}{D}\sum_{i_1, i_2, j_1, j_2=1}^d\trace{\widetilde{A}_{i_1}\widetilde{A}_{i_2}
\widetilde{A}_{j_2}^{\dagger}\widetilde{A}_{j_1}^{\dagger}}
\proji{i_1}{j_1}{E}\otimes\proji{i_2}{j_2}{E'}$. The following conclusions hold:
\begin{enumerate}
\item The operators $\rho_{EE'}$ and $\widetilde{\rho}_{EE'}$ are positive semidefinite and $\Trace{\rho_{EE'}}=\Trace{\widetilde{\rho}_{EE'}}=1$ (i.e.~they are density operators).
\item There is a local unitary $\idx{U}{E}$ such that
\begin{align*}
\norm{U_E\otimes U_{E'}\:\idx{\rho}{EE'}\:U_E^\dagger\otimes U_{E'}^\dagger-\idx{\widetilde{\rho}}{EE'}}{1}\leq 4d^2\:\norm{\cT-\widetilde{\cT}}{CB}^{1/2},
\end{align*}
where by $\Norm{\cdot}{CB}$ we denote the norm of complete boundedness \cite{Paulsen,Watrous}.
\end{enumerate}
\end{lemma}
\begin{proof}
The first assertion of the lemma follows by straightforward computations. For the second assertion we extend the CPU maps ${\cT}$ and $\widetilde{\cT}$ using Stinespring representations ${V}:=\sum_i^{d}A_i^{\dagger}\otimes\ket{i}_{E}$ and $\widetilde{V}:=\sum_i^{d}\widetilde{A}_i^{\dagger}\otimes\ket{i}_{E}$, respectively. Since
\begin{align*}
\cT(\rho)= V^{\dagger}(\rho\otimes\idi{E})V\quad\forall\rho,
\end{align*}
$V$ is indeed a Stinespring extension of $\cT$ with dilation space $\idx{\cH}{E}$. By assumption $\cT$ is unital and thus $V$ is an isometry, i.\:e. ${V}^{\dagger}V=\id$. Moreover, it is not hard to verify that the operator $\idx{\rho}{EE'}$ can be rewritten as
\begin{align*}
\idx{\rho}{EE'}=\left(\frac{1}{D}\:\ptrace{\mathbb{C}^D}{({V}\otimes\idi{E}){V}{V}^{\dagger}({V}^{\dagger}\otimes\idi{E})}\right)^{T},
\end{align*}
where $\ptrace{\mathbb{C}^D}{\cdot}$ denotes the partial trace over the output space of $\cT$ and $(\cdot)^{T}$ denotes transposition with respect to the computational basis.
The corresponding statements hold for the operators $\widetilde{\cT}$, $\widetilde{V}$ and $\idx{\widetilde{\rho}}{EE'}$. To shorten the notation we introduce the isometry $W:=(\id\otimes (U_E)^{T})V$, where $U_E$ denotes a unitary acting on the $E$ subsystem. %and write $\Delta:=W-\widetilde{V}$.
Using the above expression for $\idx{\rho}{EE'}$ and the monotonicity of the Schatten 1-norm under the partial trace, we conclude that
\begin{align*}
&\norm{U_E^\dagger\otimes U_{E'}^\dagger\:(\idx{\rho}{EE'})\:U_E\otimes U_{E'}-\idx{\widetilde{\rho}}{EE'}}{1}\\
&=\norm{(U_E\otimes U_{E'})^{T}\:(\idx{\rho}{EE'})^{T}\:(U_E^\dagger\otimes U_{E'}^\dagger)^{T}-(\idx{\widetilde{\rho}}{EE'})^{T}}{1}\\
&\leq\frac{1}{D}\Norm{({W}\otimes\idi{E}){W}{W}^{\dagger}({W}^{\dagger}\otimes\idi{E})-(\widetilde{V}\otimes\idi{E})\widetilde{V}\widetilde{V}^{\dagger}(\widetilde{V}^{\dagger}\otimes\idi{E})}{1}\\
&\leq d^2 \Norm{({W}\otimes\idi{E}){W}{W}^{\dagger}({W}^{\dagger}\otimes\idi{E})-(\widetilde{V}\otimes\idi{E})\widetilde{V}\widetilde{V}^{\dagger}(\widetilde{V}^{\dagger}\otimes\idi{E})}{\infty}\\
&\leq4d^2\:\Norm{W-\widetilde{V}}{\infty}
%&\leq \frac{d^2}{2} \left(4\norm{\Delta}{\infty}+6\norm{\Delta}%{\infty}^2+4\norm{\Delta}{\infty}^3+\norm{\Delta}{\infty}^4\right)
\end{align*}
It follows from the continuity of the Stinespring extension (see \cite[Theorem 1]{KSW06}) that the unitary $(U_{E})^{T}$ acting on the dilation space can be chosen such that
\begin{align*}
\norm{W-\widetilde{V}}{\infty}^2=\norm{(\id\otimes (U_E)^{T})V-\widetilde{V}}{\infty}^2\leq\norm{\cT-\widetilde{\cT}}{CB}.
\end{align*}
%
%Since $\norm{\cT-\widetilde{\cT}}{CB}\leq2$ the above inequality implies that
%
%\begin{align*}
%4\norm{\Delta}{\infty}+6\norm{\Delta}{\infty}^2+4\norm{\Delta}{\infty}^3+\norm{\Delta}{\infty}^4&\leq4(3+2\sqrt{2})\norm{\Delta}{\infty}\\
%&\leq4(3+2\sqrt{2})\norm{\cT-\widetilde{\cT}}{CB}^{1/2}.
%\end{align*}
%
\end{proof}
As mentioned before the operators $\rho_{EE'}$ will help us to keep track of the behaviour of the kernels of local parent Hamiltonians under scaling. The images of $\rho_{EE'}$ will correspond to the kernels of the Hamiltonians. We write $\idx{P}{EE'}$ and $\idx{\tilde{P}}{EE'}$ for the projectors onto the images of $\idx{\rho}{EE'}$ and $\idx{\tilde{\rho}}{EE'}$. In the following we shall obtain conditions under which the distance of these projectors is small, i.e.~the kernels of the parent Hamiltonians are almost the same.
\begin{lemma}\label{lem:three}
Let $\rho$ and $\tilde{\rho}$ be two Hermitian operators and let $\rho^{-1}$ and $\tilde{\rho}^{-1}$ be their (Moore-Penrose-) pseudo inverses. Let $P=\rho\rho^{-1}$ and $\tilde{P}=\tilde{\rho}\tilde{\rho}^{-1}$ denote the projectors onto the images of $\rho$ and $\tilde{\rho}$. Then for any Schatten p-norm $\norm{\cdot}{p}$ we have that
\begin{align*}
\norm{P-\tilde{P}}{p}&\leq\norm{\rho-\tilde{\rho}}{p}\left(\norm{\rho^{-1}}{\infty}+\norm{\rho^{-2}}{\infty}+\norm{\tilde{\rho}^{-2}}{\infty}+\norm{\rho^{-1}}{\infty}\norm{\tilde{\rho}^{-1}}{\infty}\right).
\end{align*}
\end{lemma}
\begin{proof} We rewrite the projectors $P$ and $\tilde{P}$ using $\rho^{-1}$ and $\tilde{\rho}^{-1}$ to conclude that
\begin{align*}
\norm{P-\tilde{P}}{p}&=
\norm{\rho\rho^{-1}-\tilde{\rho}\tilde{\rho}^{-1}-\tilde{\rho}\rho^{-1}+\tilde{\rho}\rho^{-1}}{p}\\
&\leq{\norm{\rho^{-1}}{\infty}\norm{\rho-\tilde{\rho}}{p}+\norm{\tilde{\rho}}{\infty}\norm{\rho^{-1}-\tilde{\rho}^{-1}}{p}}.
\end{align*}
The distance $\norm{\rho^{-1}-\tilde{\rho}^{-1}}{p}$ can be bounded using the fact that
\begin{align*}
\rho^{-1}-\tilde{\rho}^{-1}=\rho^{-2}(\rho-\tilde{\rho})(\id-\tilde{P})+(\id-P)(\rho-\tilde{\rho})\tilde{\rho}^{-2}-\rho^{-1}(\rho-\tilde{\rho})\tilde{\rho}^{-1}.
\end{align*}
Applying the triangle inequality and the H\"{o}lder Inequality yields
\begin{align*}
\norm{\rho^{-1}-\tilde{\rho}^{-1}}{p}\leq\norm{\rho-\tilde{\rho}}{p} \left(\norm{\rho^{-2}}{\infty}+\norm{\tilde{\rho}^{-2}}{\infty}+\norm{\rho^{-1}}{\infty}\norm{\tilde{\rho}^{-1}}{\infty}\right)
\end{align*}
which implies that
\begin{align*}
\norm{P-\tilde{P}}{p}&\leq\norm{\rho-\tilde{\rho}}{p}\left(\norm{\rho^{-1}}{\infty}+\norm{\rho^{-2}}{\infty}+\norm{\tilde{\rho}^{-2}}{\infty}+\norm{\rho^{-1}}{\infty}\norm{\tilde{\rho}^{-1}}{\infty}\right).
\end{align*}
\end{proof}
In our main derivation we will encounter the situation, where $\tilde{\rho}$ is fixed whereas $\rho$ depends on an integer, $\rho=\rho(L)$, and approaches $\tilde{\rho}$ as $L$ goes to infinity. All operators $\rho(L)$ as well as the asymptotic operator $\tilde{\rho}$ will be density operators of the same rank. We write $\mu=\mu(L)$ for the smallest non-zero eigenvalue of $\rho(L)$ and accordingly $\tilde{\mu}$ for smallest non-zero eigenvalue of $\tilde{\rho}$. By Lemma~\ref{lem:three} the convergence behaviour of the projectors $P=P(L)$ towards $\tilde{P}$ is governed by the distance $\norm{\rho-\tilde{\rho}}{p}$ and the largest eigenvalues $1/\mu$ and $1/\tilde{\mu}$ of $\rho^{-1}$ and $\tilde{\rho}^{-1}$. The upper bound for the distance between the projectors $P$ and $\tilde{P}$ obtained from Lemma~\ref{lem:three} depends explicitly on $1/\mu$. However, when $\norm{\rho-\tilde{\rho}}{\infty}$ is small enough it follows from the continuity of eigenvalues that one can replace the dependence on $1/\mu$ by $1/\tilde{\mu}$.
%
%To eliminate the dependence of the upper bound on $\mu$ we use the following lemma.
%
%
\begin{lemma}\label{lem:three(b)}
Let $\rho$ and $\tilde{\rho}$ be two density matrices of the same rank and let $\tilde{\mu}$ be the smallest positive eigenvalue of $\tilde{\rho}$. If $\norm{\rho-\tilde{\rho}}{\infty}<\tilde{\mu}$ then
\begin{align*}
\norm{P-\tilde{P}}{\infty}\leq\frac{4\norm{\rho-\tilde{\rho}}{\infty}}{(\tilde{\mu}-\norm{\rho-\tilde{\rho}}{\infty})^2}.
\end{align*}
\end{lemma}
\begin{proof}
An application of Weyl's Perturbation Theorem \cite[Corollary III.2.6]{Bhatia} under exploitation of the fact that $\rho$ and $\tilde{\rho}$ have the same rank shows that $\abs{\mu-\tilde{\mu}}\leq\norm{\rho-\tilde{\rho}}{\infty}$. This yields an upper bound on the operator norm of $\rho^{-1}$:
\begin{align*}
\norm{\rho^{-1}}{\infty}=\frac{1}{\mu}\leq\frac{1}{\tilde{\mu}-\norm{\rho-\tilde{\rho}}{\infty}}.
\end{align*}
We use Lemma~\ref{lem:three} to conclude that
\begin{align*}
\norm{P-\tilde{P}}{\infty}&\leq\norm{\rho-\tilde{\rho}}{\infty}\left(\frac{1}{\tilde{\mu}-\norm{\rho-\tilde{\rho}}{\infty}}+\frac{1}{(\tilde{\mu}-\norm{\rho-\tilde{\rho}}{\infty})^2}+\frac{1}{\tilde{\mu}^2}+\frac{1}{\tilde{\mu}(\tilde{\mu}-\norm{\rho-\tilde{\rho}}{\infty})}\right)\\
&\leq\frac{4 \norm{\rho-\tilde{\rho}}{\infty}}{(\tilde{\mu}-\norm{\rho-\tilde{\rho}}{\infty})^2}.
\end{align*}
\end{proof}
The proof of Theorem~\ref{thm1} relies on a renormalization group technique as introduced in \cite{RG}. We define local Hamiltonians acting on subchains of $\Lambda$. We then group the sites upon which these Hamiltonians act to blocks. The core observation is that the number of matrices required for the representation of the MPS will not increase from a certain point on. On the other hand with each grouping the blocked Hamiltonians \lq\lq{}look more and more classical\rq\rq{}. 
The following lemma is taken from \cite{RG} and describes this blocking procedure more precisely. The consecutive application of this result to larger and larger subchains of $\Lambda$ will be referred to as the \emph{renormalization group flow}.
\begin{lemma}\label{SVD}
Let ${\{A_i\}}_{i=1,...,d}$ be a set of $D\times D$ matrices and consider the set ${\{A_{i_1}\cdot...\cdot A_{i_L}\}}_{i_j=1,...,d}$ of all matrix products formed by matrices from ${\{A_i\}}_{i=1,...,d}$. There is a $d^L\times d^L$ unitary matrix $U$ and matrices $A^{(L)}_m$ with
\begin{align}
\sum_{i_1,...,i_L=1}^{d}{U_{m(i_1...i_L)}A_{i_1}\cdot...\cdot A_{i_L}}=A^{(L)}_m
\end{align}
such that $A^{(L)}_m = 0$ for all $m>\min{\{D^2,d^L\}}$. %Let moreover $T = \sum_{i}^d A_i\otimes\bar{A}_{i}$ and 
Moreover, it holds that $\cT^L=\cT^{(L)}$, where $\cT^{(L)}$ denotes the CP map with Kraus operators $A^{(L)}_m$.
\end{lemma}
\begin{proof} We write $(A_{i_1}\cdot...\cdot A_{i_L})_{\alpha,\beta}$ with $\alpha,\beta\in\{1,...,D\}$ for the entry of the matrix $A_{i_1}\cdot...\cdot A_{i_L}$ in row $\alpha$ and column $\beta$. 
Let $\tilde{A}$ be the $d^L\times D^2$ matrix which has the entry $(A_{i_1}\cdot...\cdot A_{i_L})_{\alpha,\beta}$ in its $(i_1...i_L)$-th row and $(\alpha,\beta)$-th column. We perform a singular value decomposition of $\tilde{A}$ writing
\begin{align*}
\tilde{A}_{(i_1...i_L),(\alpha\beta)}=\sum_{l=1}^{\min{(D^2,d^L)}}{(U^\dagger)_{(i_1...i_L),l}\:\sigma_{l}\:V_{l,(\alpha\beta)}},
\end{align*}
where by $\sigma_l$ we denote singular values. For the $m$-th row of $U\tilde{A}$, $(U\tilde{A})^{(m)}$, it holds that  
\begin{align*}
(U\tilde{A})^{(m)}=\begin{cases}\sigma_m V^{(m)}\ &;\ m\leq\min{\{d^L,D^2\}} \\ 0\ &;\ m>\min{\{d^L,D^2\}} \end{cases}.
\end{align*}
The rows of the matrix $U\tilde{A}$ correspond to the matrices $A^{(L)}_i$ and thus the first assertion of the lemma follows.
\\
For the second assertion simply observe that for any $X$ the quantity
$$\cT^L(X)=\sum_{i_i,...,i_L=1}^{d} A_{i_1}\cdot...\cdot A_{i_L}X{A}_{i_L}^\dagger\cdot...\cdot {A}_{i_1}^\dagger$$
is invariant under unitary summations i.e.
$$\cT^L(X)=\sum_{m}A_m^{(L)}X(A_m^{(L)})^\dagger=\cT^{(L)}(X).$$
\end{proof}
%
%
%We have seen that the $L$-th power of any CP map $\cT^L$ has a particularly simple structure. The associated CP map $\cT^{(L)}$ can be written with only $\min\{{d^L,D^2}\}$ Kraus operators $A^{(L)}_i$.

In the following lemma we analyse the asymptotic behaviour of the renormalization group flow and show that at large scale a generic TI MPS \lq\lq{}looks classical\rq\rq{}. To achieve this, we consider large powers of the CPU map associated to the MPS and prove that the corresponding Kraus operators have a certain structure. It is well known that condition (G1) implies that the peripheral spectrum of $\cT$ is trivial i.e.~$1$ is the only eigenvalue of $\cT$ whose magnitude is one \cite{MPS,MPSwerner,Wielandt}. 
%Diese Aussage scheint einfach aus Prop. 3 in dem Wielandt paper zu folgen. Ausserdem glaube ich dass sie direkt aus (17) im MPS paper folgt.)}
%The obvious question is now, what happens in the asymptotic limit when $L$ approaches infinity. This question was answered in \cite{RG}. Here we present a formulation of this result which we adapted to our requirements.
%
%
\begin{lemma}\label{lem:five}
Let $\cT(X)=\sum_i A_iXA_i^\dagger$ be a CPU map such that $1$ is the unique eigenvalue of magnitude one and suppose that $\Xi=\textnormal{diag}(\xi_1,...,\xi_n)$ with $\xi_i>0$ is the corresponding fixed point of $\cT^{*}$.
Then the following conclusions hold:
\begin{enumerate}
\item The limit $\cT^{\infty}:=\lim_{n\rightarrow\infty}\cT^n$ exists and we can write $\cT^{\infty}(X)=\sum_{i=1}^{D^2}A_i^{(\infty)}X(A_i^{(\infty)})^{\dagger}$ with matrices {$A^{(\infty)}_{(pq)}=\sqrt{\xi_q}\proj{p}{q}$} and $p,q\in\{1,...,D\}$.
\item The projector $P_{EE'}^{(\infty)}$ onto the image of
$$\rho_{EE'}^{(\infty)}:=\frac{1}{D}\sum_{i_1, i_2, j_1, j_2=1}^{D^2}\Trace{A_{i_1}^{(\infty)}
A_{i_2}^{(\infty)}\left(A_{j_2}^{(\infty)}\right)^{\dagger}
\left(A_{j_1}^{(\infty)}\right)^\dagger}
\proji{i_1}{j_1}{E}\otimes\proji{i_2}{j_2}{E'}$$
can be written as
\begin{align*}
P_{EE'}^{(\infty)}= \idi{A}\otimes\proji{\varphi}{\varphi}{BC}\otimes\idi{D},
\end{align*}
where $\ket{\varphi}=\sum_{i} \sqrt{\xi_i} \ket{ii}$, each of the subsystems $A,B,C,D$ is isomorphic to $\mathbb{C}^D$, and $E=AB$, $E'=CD$.
\end{enumerate}
\end{lemma}
\begin{proof}
All eigenvalues of a CPU map are contained in the closed unit disc in the complex plane. By assumption $\cT$ has only one eigenvalue on the boundary and this eigenvalue is $1$. Those eigenvalues of $\cT^n$, which are contained in the open unit disc decay with increasing $n$, while $1$ is an eigenvalue of $\cT^n$ for any $n$. Hence, $\lim_{n\rightarrow\infty}\cT^n$ simply converges to the projector onto the eigenvector $\id$ corresponding to the eigenvalue $1$ of $\cT$. %For a detailed discussion of spectral convergence estimates see \cite{szrewo}.
The fact that $A^{(\infty)}_{(pq)}=\sqrt{\xi_q}\proj{p}{q}$ is then straight forward since the dual map $(\cT^{*})^{\infty}$ acts as $(\cT^{*})^{\infty}(X)=\trace{X}\Xi$.

It follows from the first assertion of the lemma and the fact that $\{A_{i}^{(\infty)}\}_{i=1,...,D^2}$ span the entire matrix algebra that the vectors
$\ket{\mu^{(\infty)}(X)}=\sum_{i_1i_2=1}^{D^2}\trace{XA^{(\infty)}_{i_1}A^{(\infty)}_{i_2}}\ket{i_1i_2}$ span the image of $\rho_{EE'}^{(\infty)}$. Furthermore they can be written as
\begin{align*}
\ket{\mu^{(\infty)}(X)}=(\id\otimes\sqrt{\Xi}X)_{AD}\keti{\omega}{AD}\keti{\varphi}{BC},
\end{align*}
where $\ket{\omega}_{AD}=\sum_i\ket{ii}_{AD}$.
Observe that $P_{EE'}^{(\infty)}$ as defined in the lemma has rank $D^2$ and $P_{EE'}^{(\infty)}\ket{\mu^{(\infty)}(X)}=\ket{\mu^{(\infty)}(X)}$. Therefore $P_{EE'}^{(\infty)}$ projects onto the image of $\rho_{EE'}^{(\infty)}$.
\end{proof}
\subsection{The core argument}\label{coreproof}
%
%
%Let $\ket{\Psi}=\sum_{i_1,...,i_N}^{d}
%\trace{A_{i_1}\cdot...\cdot A_{i_N}}\ket{i_1...i_N}$ be a TI MPS on a ring $\Lambda$ with $N$ sites. 

In this subsection we consider the stability of the spectral gap of a parent Hamiltonian with nearest neighbour interaction $H_\Lambda=\sum_kh_{k,k+1}$. We prove that the spectral gap is stable under the assumption that at each site $\{A_i\}_{i=1,...,d}$ span the entire algebra of $D\times D$ matrices.
% matrices and that the only eigenvalue of the associated CPU map $\cT$ of magnitude one is $1$. 
In the following subsections we extend this argument to show that stability holds more generally for generic MPS in the sense of (G1).
\begin{proof}[Proof of stability (Theorem~\ref{thm1}) under strong assumptions]
We show that at large scale the parent Hamiltonian $H_\Lambda$ is a perturbation of a classical model and apply Theorem~\ref{lem:yar} to obtain the perturbation result. For this we divide $\Lambda$ into subchains $\Lambda_k$ of length $L$ and group the terms of $H_\Lambda$ into Hamiltonians $H_{\Lambda_k\cup\Lambda_{k+1}}:=\sum_{j:\{j,j+1\}\subset\Lambda_k\cup\Lambda_{k+1}}h_{j,j+1}$ acting locally on $\cH_{\Lambda_k\cup\Lambda_{k+1}}$ such that
\begin{align*}
H_\Lambda=\frac{1}{2}\sum_{k}\left(H_{\Lambda_k\cup\Lambda_{k+1}}+h_{kL,kL+1}\right).
\end{align*}
For notational convenience we shall abbreviate $H_{k,k+1}:=\frac{1}{2}\left(H_{\Lambda_k\cup\Lambda_{k+1}}+h_{kL,kL+1}\right)$. Clearly it holds that
\begin{align*}
\textnormal{Kern}\:H_{k,k+1}=\textnormal{Kern}\:H_{\Lambda_k\cup\Lambda_{k+1}}
\end{align*}
and that
\begin{align*}
H_{k,k+1}\geq\frac{1}{2} H_{\Lambda_k\cup\Lambda_{k+1}}.
\end{align*}
We introduce the density matrix
\begin{align*}
\rho_{\Lambda_k\cup\Lambda_{k+1}}:=\frac{1}{D}\sum_{i_1....i_{2L}=1\atop{j_1...j_{2L}=1}}^{d}
\trace{A_{i_1}\cdot...\cdot A_{i_{2L}} A_{j_{2L}}^\dagger\cdot...\cdot A_{j_{1}}^\dagger}\proj{i_1...i_{2L}}{j_1...j_{2L}}.
\end{align*}
By assumption the matrices $\{A_i\}_{i=1,...,d}$ span the entire matrix algebra. Hence, for any $L$ the image of $\rho_{\Lambda_k\cup\Lambda_{k+1}}$ is spanned by the $D^2$-dimensional manifold of vectors
\begin{align*}
\ket{\mu(X)}=\sum_{i_1...i_{2L}=1}^{d}\trace{XA_{i_1}\cdot...\cdot A_{i_{2L}}}\ket{i_1....i_{2L}},
\end{align*}
where $X$ is a $D\times D$ matrix with complex entries (see Section~\ref{prel:MPS}). On the other hand these vectors exactly span the kernel of $H_{\Lambda_k\cup\Lambda_{k+1}}$ (see Section~\ref{prel:parent} and \cite{MPS}) and we obtain
\begin{align*}
\textnormal{Im}\:\rho_{\Lambda_k\cup\Lambda_{k+1}}=\textnormal{Kern}\:H_{\Lambda_k\cup\Lambda_{k+1}}.
\end{align*}
%
%
%\textcolor{magenta}{Since the rank of $\rho_{\Lambda_k\cup\Lambda_{k+1}}$ (i.\:e.\:the dimension of the kernel of $H_{\Lambda_k\cup\Lambda_{k+1}}$) does not depend on $L$, $H_{\Lambda_k\cup\Lambda_{k+1}}$ is uniformly gapped. That is, there is a $\gamma$ for any $L$ such that
%
The local Hamiltonians $H_{\Lambda_k\cup\Lambda_{k+1}}$ have a positive spectral gap (see also Section~\ref{prel:parent}). Let $G_{\Lambda_k\cup\Lambda_{k+1}}$ denote the projector onto $\textnormal{Kern}\:H_{\Lambda_k\cup\Lambda_{k+1}}$ then there is a $\gamma>0$ that does not depend on $L$ such that
\begin{align}
H_{\Lambda_k\cup\Lambda_{k+1}}\geq\gamma(\id-G_{\Lambda_k\cup\Lambda_{k+1}})\label{gapped}.
\end{align}
%where $G_{\Lambda_k\cup\Lambda_{k+1}}$ is the projector onto $\textnormal{Kern}\:H_{\Lambda_k\cup\Lambda_{k+1}}$. Ok hier muss man sich nach thermodynamischem limes fragen und kontinuierlichem spektrum. trotzdem sollte die aussage stimmen.}\\
%
%
An application of Lemma~\ref{SVD} shows that there is a unitary $U_{\Lambda_k}$ acting non-trivially on $\cH_{\Lambda_k}$ only, with the property that
\begin{align}
U_{\Lambda_k}\otimes U_{\Lambda_{k+1}}\:\rho_{\Lambda_k\cup\Lambda_{k+1}}\:U_{\Lambda_k}^\dagger\otimes U_{\Lambda_{k+1}}^\dagger=\left(\begin{array}{c|c}
\:\rho_{EE'}^{(L)}\:&\:0\:
\\ \hline
\:0\:&\:0\: 
\end{array}\right),\label{crucialSVD}
\end{align}
where
\begin{align*}
\rho_{EE'}^{(L)}:=\frac{1}{D}\sum_{i_1i_2=1\atop{j_1j_2=1}}^{\min{\{D^2,d^L\}}}
\Trace{A_{i_1}^{(L)}A_{i_2}^{(L)}(A_{j_2}^{(L)})^\dagger (A_{j_1}^{(L)})^\dagger}\proji{i_1}{j_1}{E}\otimes\proji{i_2}{j_2}{E'}
\end{align*}
and the matrices $A_{i_j}^{(L)}$ are as in Lemma~\ref{SVD}. The matrix $U_{\Lambda_k}\otimes U_{\Lambda_{k+1}}\:\rho_{\Lambda_k\cup\Lambda_{k+1}}\:U_{\Lambda_k}^\dagger\otimes U_{\Lambda_{k+1}}^\dagger$ acts on a space that is isomorphic to $(\mathbb{C}^{d})^{\otimes L}\otimes(\mathbb{C}^{d})^{\otimes L}$ but only the action on a $\left(\min{\{D^2,d^L\}}\right)^2$ dimensional subspace is non-zero. In the sequel we shall assume that $L$ is chosen large such that $\rho_{EE'}^{(L)}$ acts on a $(D^2)^2$ dimensional space. For any given $L$ we fix this space and define the matrix $\rho_{\Lambda_k\cup\Lambda_{k+1}}^{(\infty)}$ by replacing $\rho_{EE'}^{(L)}$ in that space by $\rho_{EE'}^{(\infty)}$ i.e.
\begin{align*}
\rho_{\Lambda_k\cup\Lambda_{k+1}}^{(\infty)}=
\left(\begin{array}{c|c}
\:\rho_{EE'}^{(\infty)}\:&\:0\:
\\ \hline
\:0\:&\:0\: 
\end{array}\right).
\end{align*}
We denote by $G_{\Lambda_k\cup\Lambda_{k+1}}^{(\infty)}$ the projector onto the image of $\rho_{\Lambda_k\cup\Lambda_{k+1}}^{(\infty)}$. Note that since the orientation of the $(D^2)^2$ dimensional subspace in $(\mathbb{C}^d)^{\otimes L}
\otimes(\mathbb{C}^d)^{\otimes L}$ can depend on $L$ it follows that $\rho_{\Lambda_k\cup\Lambda_{k+1}}^{(\infty)}$ and $G_{\Lambda_k\cup\Lambda_{k+1}}^{(\infty)}$ can depend on $L$.\\

We will now discuss the asymptotic properties of the matrices $\rho_{\Lambda_k\cup\Lambda_{k+1}}^{(L)}$.
We will prove that with a suitable unitary transformation acting locally on the spaces $\cH_{\Lambda_k}$ and with $L$ chosen large the operators 
$\rho_{\Lambda_k\cup\Lambda_{k+1}}^{(L)}$ and $\rho_{\Lambda_k\cup\Lambda_{k+1}}^{(\infty)}$ can be made arbitrarily close. This will provide us with an explicit unitary acting locally on (sufficiently large) spaces $\cH_{\Lambda_k}$ that transforms the kernel of $H_{\Lambda_k\cup\Lambda_{k+1}}$ into a shape determined by $\rho_{\Lambda_k\cup\Lambda_{k+1}}^{(\infty)}$.

Let us consider the CPU map $\cT$ associated with the MPS $\ket{\Psi}$ and let $\lambda_2$ denote its largest in magnitude subdominant eigenvalue. We note that $\sup_{k\geq0}\Norm{\cT^k}{CB}=1$ i.e.~$\cT$ is power-bounded with respect to the $CB$-norm and constant $1$. In this situation \cite[Theorem 3.3/ Theorem 4.3]{szrewo} applies and yields an estimate for the convergence of $\cT^L$ to its stationary behaviour,
\begin{align*}
\norm{\cT^L-\cT^\infty}{CB}\leq C\abs{\lambda_2}^L.
\end{align*}
Here, $C\leq K L^{D-1}$ with $K$ that does not depend on $L$ and in generic cases $C$ does not depend on $L$, too~\cite[Theorem 4.3]{szrewo}. By Lemma~\ref{SVD} this estimate is equivalent to 
\begin{align*}
\norm{\cT^{(L)}-\cT^{(\infty)}}{CB}\leq C\abs{\lambda_2}^L,
\end{align*}
where the maps $\cT^{(L)}$ are defined in the lemma. 
%Then Lemma~\ref{lem:one} implies that
%
%\begin{align*}
%\norm{\cT^{(L)}-\cT^{(\infty)}}{CB}\leq CD\abs{\lambda_2}^L. 
%\end{align*}
%
We apply Lemma~\ref{lem:two} 
%with $\rho_{EE'}=\left(\rho_{EE'}^{(L)}\right)^{T}$ and $\tilde{\rho}_{EE'}=\left(\rho_{EE'}^{(\infty)}\right)^{T}$, where $(\cdot)^{T}$ denotes the transposition with respect to the standard basis. It follows that there is a unitary $V_E$ such that
to conclude that there is a unitary $V_E$ such that
\begin{align*}
\norm{V_E\otimes V_{E'}\:\upp{\idx{\rho}{EE'}}{L}\:V_E^\dagger\otimes V_{E'}^\dagger-\upp{\idx{\rho}{EE'}}{\infty}}{\infty}\leq4{D^4} \sqrt{C}\abs{\lambda_2}^{L/2}.
\end{align*}
By Lemma~\ref{lem:three(b)} it holds for $L$ chosen sufficiently large that
\begin{align}\label{tskfr}
\norm{V_E\otimes V_{E'}\:\upp{\idx{P}{EE'}}{L}\:V_E^\dagger\otimes V_{E'}^\dagger-\upp{\idx{P}{EE'}}{\infty}}{\infty}\leq\frac{16{D^4} \sqrt{C}\abs{\lambda_2}^{L/2}}{({\mu}-4{D^4} \sqrt{C}\abs{\lambda_2}^{L/2})^2},
\end{align}
where $\mu$ is the smallest non-zero eigenvalue of $\rho_{EE'}^{(\infty)}$. A straight forward computation shows that in fact $\mu$ equals the smallest eigenvalue of the fixed point matrix $\Xi$, see Lemma~\ref{lem:five} for the definition of $\Xi$. 

Taken together, the inequalities \eqref{tskfr} and \eqref{crucialSVD} imply that the projectors onto the images of $\rho_{\Lambda_k\cup\Lambda_{k+1}}$ and $\rho_{\Lambda_k\cup\Lambda_{k+1}}^{(\infty)}$ can be made exponentially close with a local unitary operation: There is a unitary $W_{\Lambda_k}$ such that
\begin{align}
\norm{W_{\Lambda_k}\otimes W_{\Lambda_{k+1}}G_{\Lambda_k\cup\Lambda_{k+1}} W_{\Lambda_k}^\dagger\otimes W_{\Lambda_{k+1}}^\dagger-G_{\Lambda_k\cup\Lambda_{k+1}}^{(\infty)}}{\infty}\leq \frac{16{D^4} \sqrt{C}\abs{\lambda_2}^{L/2}}{({\mu}-4{D^4} \sqrt{C}\abs{\lambda_2}^{L/2})^2}.
\label{projdistance}
\end{align}
In terms of the Hamiltonians $H_{\Lambda_k\cup\Lambda_{k+1}}$ this means that we have achieved to construct a unitary acting locally on spaces $\cH_{\Lambda_k}$ that on sufficiently large scale transforms the ground state space of $H_{\Lambda_k\cup\Lambda_{k+1}}$ into a certain subspace determined by $G_{\Lambda_k\cup\Lambda_{k+1}}^{(\infty)}$.
In the next step we construct a classical Hamiltonian with this ground state subspace. For each $L$ the structure of the operators $G_{\Lambda_k\cup\Lambda_{k+1}}^{(\infty)}$ is known from Lemma~\ref{lem:five}. We have that
\begin{align*}
G_{\Lambda_k\cup\Lambda_{k+1}}^{(\infty)}= \left(\begin{array}{c|c}
\:\idi{A}\otimes\proji{\varphi}{\varphi}{BC}\otimes\idi{D}\:&\:0\:
\\ \hline
\:0\:&\:0\: 
\end{array}\right)
\end{align*}
with $\ket{\varphi}=\sum_{i} \sqrt{\xi_i} \ket{ii}$. Thus $G_{\Lambda_k\cup\Lambda_{k+1}}^{(\infty)}$ induces a natural decomposition of $\cH_{\Lambda_k\cup\Lambda_{k+1}}$ into a subspace $\cH_{X}$ on which $G_{\Lambda_k\cup\Lambda_{k+1}}^{(\infty)}$ acts as the zero operator and a subspace which is isomorphic to $\mathbb{C}^{D^2}\otimes\mathbb{C}^{D^2}$. The latter can further be decomposed according to the structure of $G_{\Lambda_k\cup\Lambda_{k+1}}^{(\infty)}$ into $\mathbb{C}^{D^2}\otimes\mathbb{C}^{D^2}\cong\mathbb{C}^{D}_{A}\otimes\mathbb{C}^{D}_{B}\otimes\mathbb{C}^{D}_{C}\otimes\mathbb{C}^{D}_{D}$. By an additional decomposition of $\cH_X$ and choosing $L$ even we achieve the decomposition
\begin{align*}
\cH_{\Lambda_k\cup\Lambda_{k+1}}\cong(\mathbb{C}^{D}_{A}\oplus\cH_{X_A})\otimes(\mathbb{C}^{D}_{B}\oplus\cH_{X_B})\otimes(\mathbb{C}^{D}_{C}\oplus\cH_{X_C})\otimes(\mathbb{C}^{D}_{D}\oplus\cH_{X_D}).
\end{align*}
Here the spaces $\cH_{X_A},...,\cH_{X_D}$ are chosen to have dimension $d^{L/2}-D$. In the decomposition of $\cH_{\Lambda_k\cup\Lambda_{k+1}}$ we identify the ``half-shifted'' spaces $\cH_{\Lambda_k\cup\Lambda_{k+1}}^{\textnormal{HS}}
:=(\mathbb{C}^{D}_{B}\oplus\cH_{X_B})\otimes(\mathbb{C}^{D}_{C}\oplus\cH_{X_C})$. Figure~\ref{decompose} shows a schematic representation of this decomposition.
Note that $\cH_{\Lambda_k\cup\Lambda_{k+1}}^{\textnormal{HS}}
\cong\cH_{\Lambda_k}$ and that the following inclusions hold:
\begin{align*}
\cH_{\Lambda_k\cup\Lambda_{k+1}}^{\textnormal{HS}}
\subset\cH_{\Lambda_k\cup\Lambda_{k+1}}
\subset\cH_{\Lambda_{k-1}\cup\Lambda_{k}}
^{\textnormal{HS}}\otimes
\cH_{\Lambda_k\cup\Lambda_{k+1}}
^{\textnormal{HS}}
\otimes
\cH_{\Lambda_{k+1}\cup\Lambda_{k+2}}^{\textnormal{HS}}.
\end{align*}
Let $H_{\Lambda_k\cup\Lambda_{k+1}}^{\textnormal{HS}}$ denote the projector in $\cH_{\Lambda_k\cup\Lambda_{k+1}}^{\textnormal{HS}}$ onto the orthogonal complement of $\ket{\varphi}$. The above inclusions translate into the estimates
\begin{align}
H_{\Lambda_k\cup\Lambda_{k+1}}^{\textnormal{HS}}\leq\id-G_{\Lambda_k\cup\Lambda_{k+1}}^{(\infty)}\leq H_{\Lambda_{k-1}\cup\Lambda_{k}}^{\textnormal{HS}}
+H_{\Lambda_k\cup\Lambda_{k+1}}^{\textnormal{HS}}
+H_{\Lambda_{k+1}\cup\Lambda_{k+2}}^{\textnormal{HS}}.\label{hvbounds}
\end{align}

\begin{figure}
\centering
\begin{tikzpicture}[block/.style = {align=center, anchor=center, minimum height=1cm, text width=1.5cm, inner sep=1}, brace/.style={decoration={brace},    decorate}]

\tikzstyle{abstract}=[rectangle, draw=black, rounded corners, line width=0.5mm, minimum  width=.5cm, minimum height=2cm, anchor=center, text width=.5cm]
\tikzstyle{abstracth}=[rectangle, draw=black, rounded corners, line width=0.5mm, minimum  width=.7cm, minimum height=0.7cm, anchor=center, text width=.5cm]
\tikzstyle{txx}=[rectangle, line width=0.0mm, minimum width=2cm, minimum height=0.3cm, text centered, anchor=center, text width=1cm]
\tikzstyle{brox}=[rectangle, draw=black, line width=0.2mm, minimum  width=4cm, minimum height=3cm]

\node[txx,scale=0.7] (dec) at (5.75+5.25,0.35) {$\bullet\hskip0.23cm \bullet\hskip0.23cm\bullet$};

\node[txx,scale=0.7] (dec) at (5.75-5.25,0.35) {$\bullet\hskip0.23cm \bullet\hskip0.23cm\bullet$};

\node[brox] (dec) at (5.75,0.35) {};

\node[abstract] (dec) at (2,0) {};
\node[abstracth] (dec) at (2,1.35) {};

\node[abstract] (dec) at (3,0) {};
\node[abstracth] (dec) at (3,1.35){};

\node[abstract] (dec) at (4.25,0) {};
\node[scale=0.9, text centered] (dec) at (4.27,0) {$\cH_{X_{A}}$};
\node[abstracth] (dec) at (4.25,1.35) {$\mathbb{C}^D_A$};

\node[abstract] (dec) at (5.25,0) {};
\node[scale=0.9, text centered] (dec) at (5.27,0) {$\cH_{X_{B}}$};
\node[abstracth] (dec) at (5.25,1.35) {$\mathbb{C}^D_B$};

\node[abstract] (dec) at (6.25,0) {};
\node[scale=0.9, text centered] (dec) at (6.27,0) {$\cH_{X_{C}}$};
\node[abstracth] (dec) at (6.25,1.35) {$\mathbb{C}^D_C$};

\node[abstract] (dec) at (7.25,0) {};
\node[scale=0.9, text centered] (dec) at (7.27,0) {$\cH_{X_{D}}$};
\node[abstracth] (dec) at (7.25,1.35) {$\mathbb{C}^D_D$};

\node[abstract] (dec) at (8.5,0) {};
\node[abstracth] (dec) at (8.5,1.35) {};

\node[abstract] (dec) at (9.5,0) {};
\node[abstracth] (dec) at (9.5,1.35) {};

%beschriftung unten untere teil

\node[block,anchor=center] (nop) at
   (1.5,-2.25) {\large $\cH_{\Lambda_{k-2}\cup\Lambda_{k-1}}$};

\node[block,anchor=center] (nop) at
   (5.75,-2.25) {\large $\cH_{\Lambda_{k}\cup\Lambda_{k+1}}$};
   
\node[block,anchor=center] (nop) at
   (10.00,-2.25) {\large $\cH_{\Lambda_{k+2}\cup\Lambda_{k+3}}$};

%beschriftung oberer teil

\node[block,anchor=center] (nop) at
   (5.75,2.4) {\large $\cH^{\textnormal{HS}}_{\Lambda_{k}\cup\Lambda_{k+1}}$};

\node[block,anchor=center] (nop) at
   (3.75,2.4) {\large $\cH^{\textnormal{HS}}_{\Lambda_{k-1}\cup\Lambda_{k}}$};

\node[block,anchor=center] (nop) at
   (7.75,2.4) {\large $\cH^{\textnormal{HS}}_{\Lambda_{k+1}\cup\Lambda_{k+2}}$};

%beschriftung unten oberer teil

\node[block,anchor=center] (nop) at
   (2.5,-1.65) {\large $\cH_{\Lambda_{k-1}}$};
   
\node[block,anchor=center] (nop) at
   (4.75,-1.65) {\large $\cH_{\Lambda_{k}}$};
   
\node[block,anchor=center] (nop) at
   (6.75,-1.65) {\large $\cH_{\Lambda_{k+1}}$};

\node[block,anchor=center] (nop) at
   (9,-1.65) {\large $\cH_{\Lambda_{k+2}}$};

%% lower downpart nodes for braces

\node[block,anchor=center,text width=0] (bA) at
   (3.5,-1.35) {};
\node[block,anchor=center,text width=0] (bE) at
   (-.5,-1.35) {};
   
\node[block,anchor=center,text width=0] (cA) at
   (7.75,-1.35) {};
\node[block,anchor=center,text width=0] (cE) at
   (3.75,-1.35) {};
   
\node[block,anchor=center,text width=0] (dA) at
   (12,-1.35) {};
\node[block,anchor=center,text width=0] (dE) at
   (8,-1.35) {};

%upper downpart nodes for braces
   
\node[block,anchor=center,text width=0] (bAO) at
   (3.4,-0.75) {};
\node[block,anchor=center,text width=0] (bEO) at
   (1.6,-0.75) {};

\node[block,anchor=center,text width=0] (cEO) at
   (3.85,-0.75) {};
\node[block,anchor=center,text width=0] (cAO) at
   (5.65,-0.75) {};

\node[block,anchor=center,text width=0] (dEO) at
   (5.85,-0.75) {};
\node[block,anchor=center,text width=0] (dAO) at
   (7.65,-0.75) {};
   
\node[block,anchor=center,text width=0] (eEO) at
   (8.1,-0.75) {};
\node[block,anchor=center,text width=0] (eAO) at
   (9.9,-0.75) {};

%nodes für upper braces

\node[block,anchor=center,text width=0] (AOO) at
   (4.85,2.5) {};
\node[block,anchor=center,text width=0] (BOO) at
   (6.65,2.5) {};

\node[block,anchor=center,text width=0] (COO) at
   (3.65+1,2.5) {};
\node[block,anchor=center,text width=0] (DOO) at
   (3.65-1,2.5) {};

\node[block,anchor=center,text width=0] (EOO) at
   (7.85+1,2.5) {};
\node[block,anchor=center,text width=0] (FOO) at
   (7.85-1,2.5) {};

%braces alle

\draw [brace] (bA.south) -- (bE.south);
\draw [brace] (cA.south) -- (cE.south);
\draw [brace] (dA.south) -- (dE.south);

\draw [brace] (bAO.south) -- (bEO.south);
\draw [brace] (cAO.south) -- (cEO.south);
\draw [brace] (dAO.south) -- (dEO.south);
\draw [brace] (eAO.south) -- (eEO.south);

\draw [brace] (AOO.south) -- (BOO.south);
\draw [brace] (DOO.south) -- (COO.south);
\draw [brace] (FOO.south) -- (EOO.south);

%\node[block,anchor=center,text width=4cm] (f) at
 %  (0,-0.65) {$\cH_\Lambda=\bigotimes_{x\in\Lambda}\cH_x$};

%\node[block,anchor=center,text width=0] (bX) at
 %  (2,-6) {};
%\node[block,anchor=center,text width=0] (bY) at
  % (4,-6) {};   
%\draw [brace] (bX.north) -- (bY.north);

\end{tikzpicture}

\caption{Scheme of Hilbert space structure of a segment of the grouped spin chain. Boxes with round corners depict Hilbert spaces; separated boxes are tensored while merged boxes mean a direct sum. Large rectangular box in the middle shows decomposition of $\cH_{\Lambda_k\cup\Lambda_{k+1}}$ into four subspaces. Half-shifted spaces $\cH^{\textnormal{HS}}$ are identified at top of the scheme. Dots on left- and right-hand side denote periodic continuation of Hilbert space structure.\label{decompose}}
\end{figure}
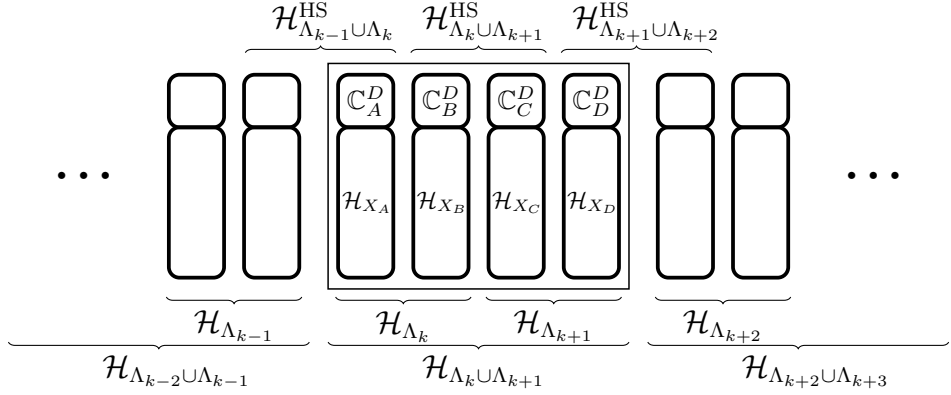

Consider the operator
\begin{align*}
H_{\Lambda}^{\textnormal{CL}}:=3L\sum_{k}H_{\Lambda_k\cup\Lambda_{k+1}}^{\textnormal{HS}}.
\end{align*}
This operator is classical in the sense of Theorem~\ref{lem:yar} with respect to the half-shifted spaces $H_{\Lambda_k\cup\Lambda_{k+1}}^{\textnormal{HS}}$. We claim that for $L$ chosen large enough $(\bigotimes_k W_{\Lambda_k})H_\Lambda(\bigotimes_k W_{\Lambda_k})^\dagger$ is a perturbation of $H_{\Lambda}^{\textnormal{CL}}$ satisfying the assumptions of Theorem~\ref{lem:yar}. We construct this perturbation explicitly. It consists of a bounded part
\begin{align*}
\phi^{(b)}_{k,k+1}:=&W_{\Lambda_k}\otimes W_{\Lambda_{k+1}}(\id-G_{\Lambda_k\cup\Lambda_{k+1}})H_{k,k+1}(\id-G_{\Lambda_k\cup\Lambda_{k+1}})W_{\Lambda_k}^\dagger\otimes W_{\Lambda_{k+1}}^\dagger\\
&-(\id-G_{\Lambda_k\cup\Lambda_{k+1}}^{(\infty)}) W_{\Lambda_{k}}\otimes W_{\Lambda_{k+1}} H_{k,k+1}W_{\Lambda_k}^\dagger\otimes W_{\Lambda_{k+1}}^\dagger (\id-G_{\Lambda_k\cup\Lambda_{k+1}}^{(\infty)})
\end{align*}
and a relatively bounded part
\begin{align*}
\phi^{(r)}_{k,k+1}:=&(\id-G_{\Lambda_k\cup\Lambda_{k+1}}^{(\infty)}) W_{\Lambda_{k}}\otimes W_{\Lambda_{k+1}} H_{k,k+1}W_{\Lambda_k}^\dagger\otimes W_{\Lambda_{k+1}}^\dagger (\id-G_{\Lambda_k\cup\Lambda_{k+1}}^{(\infty)})\\
&-L\:(H_{\Lambda_{k-1}\cup\Lambda_{k}}^{\textnormal{HS}}+
H_{\Lambda_k\cup\Lambda_{k+1}}^{\textnormal{HS}}+
H_{\Lambda_{k+1}\cup\Lambda_{k+2}}^{\textnormal{HS}}).
\end{align*}
Taking both together yields
\begin{align*}
(\bigotimes_k W_{\Lambda_k})H_\Lambda (\bigotimes_k W_{\Lambda_k})^\dagger=H_{\Lambda}^{\textnormal{CL}}+\sum_k\phi^{(b)}_{k,k+1}+\sum_k\phi^{(r)}_{k,k+1}.
\end{align*}
First we estimate
\begin{align*}
\norm{\phi^{(b)}_{k,k+1}}{\infty}=
&\Big|\!\Big|W_{\Lambda_k}\otimes W_{\Lambda_{k+1}}(\id-G_{\Lambda_k\cup\Lambda_{k+1}})H_{k,k+1}(\id-G_{\Lambda_k\cup\Lambda_{k+1}})W_{\Lambda_k}^\dagger\otimes W_{\Lambda_{k+1}}^\dagger\\
&-(\id-G_{\Lambda_k\cup\Lambda_{k+1}}^{(\infty)}) W_{\Lambda_{k}}\otimes W_{\Lambda_{k+1}} H_{k,k+1}W_{\Lambda_k}^\dagger\otimes W_{\Lambda_{k+1}}^\dagger (\id-G_{\Lambda_k\cup\Lambda_{k+1}}^{(\infty)})\Big|\!\Big|_{\infty}\\
%&\leq\Norm{W\otimes W G_{\Lambda_k\cup\Lambda_{k+1}} W^\dagger\otimes W^\dagger-G_{\Lambda_k\cup\Lambda_{k+1}}^{(\infty)}}{\infty}\cdot\\
%&\Bigg(\Norm{W\otimes W H_{k,k+1}W^\dagger\otimes W^\dagger\ \left(\id-W\otimes W G_{\Lambda_k\cup\Lambda_{k+1}}W^\dagger\otimes W^\dagger\right)}{\infty}+\\
%&\qquad\Norm{(\id-G_{\Lambda_k\cup\Lambda_{k+1}}^{(\infty)})\ W\otimes W H_{k,k+1}W^\dagger\otimes W^\dagger}{\infty}\Bigg)\\
&\leq\Norm{H_{k,k+1}(\id-G_{\Lambda_k\cup\Lambda_{k+1}})W_{\Lambda_k}^\dagger\otimes W_{\Lambda_{k+1}}^\dagger-H_{k,k+1}W_{\Lambda_k}^\dagger\otimes W_{\Lambda_{k+1}}^\dagger (\id-G_{\Lambda_k\cup\Lambda_{k+1}}^{(\infty)})}{\infty}\\
&+\Norm{W_{\Lambda_k}\otimes W_{\Lambda_{k+1}}(\id-G_{\Lambda_k\cup\Lambda_{k+1}})H_{k,k+1}-(\id-G_{\Lambda_k\cup\Lambda_{k+1}}^{(\infty)}) W_{\Lambda_{k}}\otimes W_{\Lambda_{k+1}} H_{k,k+1}}{\infty}\\
&\leq2\:\Norm{W\otimes W G_{\Lambda_k\cup\Lambda_{k+1}} W^\dagger\otimes W^\dagger-G_{\Lambda_k\cup\Lambda_{k+1}}^{(\infty)}}{\infty}\:\Norm{H_{k,k+1}}{\infty}\\
&\leq\frac{32L{D^4} \sqrt{C}\abs{\lambda_2}^{L/2}}{({\mu}-4{D^4} \sqrt{C}\abs{\lambda_2}^{L/2})^2}.
\end{align*}
The last inequality makes use of \eqref{projdistance} and the fact that $\norm{H_{k,k+1}}{\infty}\leq L$. Thus we have shown that the norm of $\phi^{(b)}_{k,k+1}$ decays exponentially fast with increasing size of the blocks $\Lambda_k$.\\
To verify that $\phi^{(r)}_{x}$ is in accordance with the conditions of Theorem~\ref{lem:yar}  we need to estimate $\abs{\sum_{x\in I}\phi^{(r)}_{x}}$
for any $I\subset\{1,...,N/L\}$. The maximum is attained when $I=\Lambda$ since
\begin{align*}
\phi^{(r)}_{k,k+1}&\leq L\:(\id-G_{\Lambda_k\cup\Lambda_{k+1}}^{(\infty)})-L\:(H_{\Lambda_{k-1}\cup\Lambda_{k}}^{\textnormal{HS}}+
H_{\Lambda_k\cup\Lambda_{k+1}}^{\textnormal{HS}}+
H_{\Lambda_{k+1}\cup\Lambda_{k+2}}^{\textnormal{HS}})\\
&\leq0,
\end{align*}
where the second inequality makes use of \eqref{hvbounds}. A lower bound on $\phi^{(r)}_{k,k+1}$ follows from the gappedness of $H_{k,k+1}$ \eqref{gapped}:
\begin{align*}
&(\id-G_{\Lambda_k\cup\Lambda_{k+1}}^{(\infty)}) W_{\Lambda_{k}}\otimes W_{\Lambda_{k+1}} H_{k,k+1}W_{\Lambda_k}^\dagger\otimes W_{\Lambda_{k+1}}^\dagger (\id-G_{\Lambda_k\cup\Lambda_{k+1}}^{(\infty)})\geq\\
&\frac{\gamma}{2}\:(\id-G_{\Lambda_k\cup\Lambda_{k+1}}^{(\infty)}) W_{\Lambda_{k}}\otimes W_{\Lambda_{k+1}} (\id-G_{\Lambda_k\cup\Lambda_{k+1}})W_{\Lambda_k}^\dagger\otimes W_{\Lambda_{k+1}}^\dagger (\id-G_{\Lambda_k\cup\Lambda_{k+1}}^{(\infty)})\geq\\
&\frac{\gamma}{2}\:\left(\id-G_{\Lambda_k\cup\Lambda_{k+1}}^{(\infty)}\right) \left(\id-G_{\Lambda_k\cup\Lambda_{k+1}}^{(\infty)}-\frac{16{D^4} \sqrt{C}\abs{\lambda_2}^{L/2}}{({\mu}-4{D^4} \sqrt{C}\abs{\lambda_2}^{L/2})^2}\cdot\id\right) \left(\id-G_{\Lambda_k\cup\Lambda_{k+1}}^{(\infty)}\right)\geq\\
&\frac{\gamma}{2}\:\left(1-\frac{16{D^4} \sqrt{C}\abs{\lambda_2}^{L/2}}{({\mu}-4{D^4} \sqrt{C}\abs{\lambda_2}^{L/2})^2}\right)\:\left(\id-G_{\Lambda_k\cup\Lambda_{k+1}}^{(\infty)}\right)\geq
\frac{\gamma}{2}\:\left(1-\frac{16{D^4} \sqrt{C}\abs{\lambda_2}^{L/2}}{({\mu}-4{D^4} \sqrt{C}\abs{\lambda_2}^{L/2})^2}\right)
\:H_{\Lambda_k\cup\Lambda_{k+1}}^{\textnormal{HS}}.
\end{align*}
We sum the terms $\phi^{(r)}_{k,k+1}$ to conclude that
\begin{align*}
&\sum_{k,k+1}\phi^{(r)}_{k,k+1}\geq\\
&\sum_{k,k+1}\left(\frac{\gamma}{2}\:\left(1-\frac{16{D^4} \sqrt{C}\abs{\lambda_2}^{L/2}}{({\mu}-4{D^4} \sqrt{C}\abs{\lambda_2}^{L/2})^2}\right)
\:H_{\Lambda_k\cup\Lambda_{k+1}}^{\textnormal{HS}}-L\:(H_{\Lambda_{k-1}\cup\Lambda_{k}}^{\textnormal{HS}}+
H_{\Lambda_k\cup\Lambda_{k+1}}^{\textnormal{HS}}+
H_{\Lambda_{k+1}\cup\Lambda_{k+2}}^{\textnormal{HS}})\right)\\
&=\left(-1+\frac{\gamma}{6L}-\frac{8\gamma{D^4} \sqrt{C}\abs{\lambda_2}^{L/2}}{3L({\mu}-4{D^4} \sqrt{C}\abs{\lambda_2}^{L/2})^2}\right)
\:H_{\Lambda}^{\textnormal{CL}}. 
\end{align*}
Thus for Theorem~\ref{lem:yar} we have that
\begin{align*}
\Abs{\sum_{k,k+1}\bracket{\psi}{\phi^{(r)}_{k,k+1}}{\psi}}\leq\alpha\:\bracket{\psi}{H_{\Lambda}^{\textnormal{CL}}}{\psi}
\end{align*}
with $\alpha=(1-\frac{\gamma}{6L}+\cO(\abs{\lambda_2}^{\frac{L}{2}}))$ and $\beta=\delta\:(\frac{\gamma}{6L}-\cO(\abs{\lambda_2}^{\frac{L}{2}}))^{2\kappa}$, where the constants $\delta$ and $\kappa$ still have to be chosen appropriately. As long as $\gamma$ decays sub-exponentially fast with $L$, for $L$ sufficiently large $\norm{\phi^{(b)}_{k,k+1}}{\infty}\leq\beta$ holds. For parent Hamiltonians, which have a constant local gap this is certainly the case.

Applying Theorem~\ref{lem:yar} we could recover the well-known fact that $H_\Lambda$ has a gapped ground state. However, the conditions of Theorem~\ref{lem:yar} are ``open'' in the sense that adding sufficiently small bounded perturbation to $\phi^{(b)}_{k,k+1}$ still results in a total perturbation, which is within the range where Theorem~\ref{lem:yar} can be applied. This provides us with a perturbation result for Hamiltonians in the neighbourhood of $H_\Lambda$. More precisely, let $\Phi':=\sum_{k,k+1}\phi_{k,k+1}'$ be a finite range interaction with $\norm{\phi_{k,k+1}'}{\infty}\leq\beta'$ and $\beta'>0$ small enough. We analyse the spectral gap of $H_\Lambda'=H_\Lambda+\Phi'$. Suppose for the moment that $\phi_{k,k+1}'$ acts exactly on $\cH_{\Lambda_k\cup\Lambda_{k+1}}$ and let
\begin{align*}
\phi_{k,k+1}'':=W_{\Lambda_k}\otimes W_{\Lambda_{k+1}}\phi_{k,k+1}' W_{\Lambda_k}^\dagger\otimes W_{\Lambda_{k+1}}^\dagger.
\end{align*}
Consider the Hamiltonian
\begin{align*}
(\bigotimes_k W_{\Lambda_k})H_\Lambda (\bigotimes_k W_{\Lambda_k})^\dagger+\sum_{k}\phi_{k,k+1}''=(\bigotimes_k W_{\Lambda_k})(H_\Lambda+\Phi') (\bigotimes_k W_{\Lambda_k})^\dagger.
\end{align*}
If $\beta'>0$ is chosen sufficiently small Theorem~\ref{lem:yar} applies and proves the stability of the spectral gap of $H_\Lambda+\Phi'$. In general, though, we want to allow an arbitrary (finite) interaction range for $\phi_{k,k+1}$. If $\phi_{k,k+1}$ acts nontrivially on a subchain of $\Lambda_k\cup\Lambda_{k+1}$ only it is possible to group the $\phi_{k,k+1}$ terms in such a way that in total one gets a finite range interaction on $\Lambda_k\cup\Lambda_{k+1}$. Choosing $\beta'$ we make sure that the grouped perturbation is sufficiently small for an application of Theorem~\ref{lem:yar}. On the other hand if the perturbation has interaction range exceeding the subchain $\Lambda_k\cup\Lambda_{k+1}$ one simply chooses $L$ larger and the previous discussion applies to the larger subchains.
\end{proof}
\subsection{Proof of Theorem~\ref{thm1} and Corollary~\ref{non-generic}}\label{finishproof}
\begin{proof}[Proof of Theorem~\ref{thm1}]
The proof is a simple upgrade of the restricted discussion of the previous subsection. By condition (G1) there is finite $P_0$ such that the matrices $\{A_{i_1}\cdot...\cdot A_{i_{P_0}}\}$ span the whole algebra of $D\times D$ matrices. Hence, $H_\Lambda=\sum_i\tau^i(h_{G_P})$ has a unique ground state  for any $P>P_0$, see Section~\ref{prel:parent}. We proceed as in the proof of the theorem and divide $\Lambda$ into chains $\Lambda_k$ of length $L$. In addition we assume that the chains are sufficiently large to support $h_{\cG_P}$ i.\:e.\: $L\geq P$.
We define the operators 
\begin{align*}
H_{\Lambda_k\cup\Lambda_{k+1}}:=\sum_{i : \{i+1,...,i+P\}\subset\Lambda_k\cup\Lambda_{k+1}}\tau^i(h_{\cG_P}),
\end{align*}
which are sums of all the translates of $h_{\cG_P}$ that act locally on $\Lambda_k\cup\Lambda_{k+1}$.
There are $P-1$ terms in the above Hamiltonian that partially act on block $\Lambda_k$ and partially on $\Lambda_{k+1}$. We define the operators $H_{k,k+1}$ by adding these terms to $H_{\Lambda_k\cup\Lambda_{k+1}}$. Formally
\begin{align*}
H_{k,k+1}= \frac{1}{2}H_{\Lambda_k\cup\Lambda_{k+1}}+\frac{1}{2}\sum_{i: (i+1\in\Lambda_k\ \wedge\ i+P\in\Lambda_{k+1})} \tau^i(h_{\cG_P}).
\end{align*}
As before, we have the properties
\begin{align*}
&H_{k,k+1}\geq H_{\Lambda_k\cup\Lambda_{k+1}},\\
&\textnormal{Kern}(H_{k,k+1})=\textnormal{Kern}(H_{\Lambda_k\cup\Lambda_{k+1}})
\end{align*}
and
\begin{align*}
H_{\Lambda}=\sum_k H_{k,k+1}.
\end{align*}
The kernel of $H_{\Lambda_k\cup\Lambda_{k+1}}$ is given by the image (see also \cite[Section~4.1.1]{MPS}) of
\begin{align*}
\rho_{\Lambda_k\cup\Lambda_{k+1}}=\sum_{i_1....i_{2L}=1\atop{j_1...j_{2L}=1}}^{d}\trace{A_{i_1}\cdot...\cdot A_{i_{2L}}A_{j_{2L}}^\dagger\cdot...\cdot A_{j_{1}}^\dagger}\proj{i_1...i_{2L}}{j_1...j_{2L}}.
\end{align*}
As before, the spectral gap of $H_{\Lambda_k\cup\Lambda_{k+1}}$ can be lower bounded by some constant. With $G_{\Lambda_k\cup\Lambda_{k+1}}$ and $G_{\Lambda_k\cup\Lambda_{k+1}}^{(\infty)}$ defined as in Subsection~\ref{coreproof} the derivation follows the same lines as before. Hence, stability follows under condition (G1), which completes the proof of Theorem~\ref{thm1}.
%
%\begin{align*}
%G_{\Lambda_k\cup\Lambda_{k+1}}^{(\infty)}= \left(\begin{array}{c|c}
%\:\idi{A}\otimes\proji{\varphi}{\varphi}{BC}\otimes\idi{D}\:&\:0\:
%\\ \hline
%\:0\:&\:0\: 
%\end{array}\right)
%\end{align*}
%with a local unitary transformation. Choosing $L$ even, the space $\cH_{\Lambda_k\cup\Lambda_{k+1}}$ can be decomposed according to the structure of $G_{\Lambda_k\cup\Lambda_{k+1}}^{(\infty)}$ into
%\begin{align*}
%\cH_{\Lambda_k\cup\Lambda_{k+1}}\cong(\mathbb{C}^{D}_{A}\oplus\cH_{X_A})\otimes(\mathbb{C}^{D}_{B}\oplus\cH_{X_B})\otimes(\mathbb{C}^{D}_{C}\oplus\cH_{X_C})\otimes(\mathbb{C}^{D}_{D}\oplus\cH_{X_D}).
%\end{align*}
%The classical Hamiltonian $H_{\Lambda,0}$ and the perturbation $\phi^{(b)}_{k,k+1}$ and $\phi^{(r)}_{k,k+1}$ are defined as in the proof of Theorem~\ref{thm1}, where the bounds on $\phi^{(b)}_{k,k+1}$ and $\sum_k\phi^{(r)}_{k,k+1}$ remain valid. The conclusions of the Corollary then follow by an application of Lemma~\ref{lem:yar} similar to the discussion in the proof of Theorem~\ref{thm1}.
\end{proof}
\begin{proof}[Proof of Corollary~\ref{non-generic}]%
As before we choose $L\geq P$ and divide $\Lambda$ into subchains of length $L$. The restrictions of $\hat{H}_\Lambda=\sum_jh_{j,j+1}$ and $H_\Lambda=\sum_i\tau^i(h_{G_P})$ to $\Lambda_k\cup\Lambda_{k+1}$ are given by
$\hat{H}_{\Lambda_k\cup\Lambda_{k+1}}=
\sum_{\{j,j+1\}\subset\Lambda_k\cup\Lambda_{k+1}}h_{j,j+1}$ and $H_{\Lambda_k\cup\Lambda_{k+1}}=\sum_{i : \{i+1,...,i+P\}\subset\Lambda_k\cup\Lambda_{k+1}}\tau^i(h_{\cG_P})$. The condition $c_1\:h_{G_P}\leq\sum_{j=1}^{P-1}h_{j,j+1}\leq c_2\:h_{G_P}$ implies that
\begin{align*}
c_1\:H_{\Lambda_k\cup\Lambda_{k+1}}\leq\sum_{i : \{i+1,...,i+P\}\subset\Lambda_k\cup\Lambda_{k+1}}\tau^i\left(\sum_{j=1}^{P-1}h_{j,j+1}\right)\leq c_2\:H_{\Lambda_k\cup\Lambda_{k+1}}.
\end{align*}
It follows that $\sum_{i : \{i+1,...,i+P\}\subset\Lambda_k\cup\Lambda_{k+1}}\tau^i\left(\sum_{j=1}^{P-1}h_{j,j+1}\right)$ has the same kernel as $H_{\Lambda_k\cup\Lambda_{k+1}}$.
Thus the kernels of $H_{\Lambda_k\cup\Lambda_{k+1}}$ and $\hat{H}_{\Lambda_k\cup\Lambda_{k+1}}$ are identical and Corollary~\ref{non-generic} follows from the derivation of Theorem~\ref{thm1}.
\end{proof}
\section{Discussion}
We have proven that in the generic case (G1) the ground state of the parent Hamiltonian
model is stable under sufficiently small perturbations of bounded interaction range. A core point in our approach lies in the construction of a renormalization group flow that converges to an essentially classical model. 
This method provides an independent proof and a clear physical intuition for stability of parent Hamiltonians: \lq\lq{}The parent Hamiltonian model is stable because on sufficiently large scale it is essentially classical\rq\rq{}. Given the previous proof of stability~\cite{Stable,Rob} the core innovation of this article can be seen in this observation. To illustrate how our new intuition can be useful we briefly discuss the classification of quantum phases of gapped systems with MPS ground states. In a recent publication~\cite{MPSph1} it was shown that in the absence of symmetry protection all such systems are in the same phase, up to possible ground-state degeneracies. The techniques developed in our article yield an \lq\lq{}immediate\rq\rq{} proof of this fact, while at the core of the approach of~\cite{MPSph1} lies the so-called \emph{isometric form}, which is a new standard form for MPS. More precisely, two translationally invariant,
gapped, local Hamiltonians $H^{(p)}$ with $p\in\{0,1\}$ are defined to be in same phase~\cite[Section II.C.1]{MPSph1} iff there exists a finite block length $K$ such that after grouping $K$ sites $H^{(p)}$ are two-local and there exists a translationally invariant path of tow-local Hamiltonians $h^{(\gamma)}_{k,k+1}$ with

\emph{i)} $h^{(\gamma=0)}=h^{(0)}$ and $h^{(\gamma=1)}=h^{(1)}$,

\emph{ii)} $\Norm{h^{(\gamma)}}{}\leq1$,

\emph{iii)} $h^{(\gamma)}$ depends continuously on $\gamma$ and

\emph{iv)} $H^{(\gamma)}=\sum_k h^{(\gamma)}_{k,k + 1}$ has a spectral gap above the ground state manifold.

From the preceding discussion we expect that according to this definition a gapped parent Hamiltonian $H$ and the corresponding classical Hamiltonian $H^{\textnormal{CL}}$ are in the same phase. This is clearly the case as we can continuously switch on the perturbations $\sum_k\phi^{(b)}_{k,k+1}$ and $\sum_k\phi^{(r)}_{k,k+1}$ to obtain a gapped Hamiltonian path interpolating between $H$ and $H^{\textnormal{CL}}$. Furthermore, it is clear that any classical states share the same phase~\cite{MPSph1} and~\cite{MPSph2}, which proves the mentioned result.
\begin{corollary}[\cite{MPSph1}]\label{keinBock}
Let $H^{(0)}_\Lambda$ and $H^{(1)}_\Lambda$ be TI Hamiltonians on a ring $\Lambda$ with PBC and suppose that $H^{(0)}_\Lambda$ and $H^{(1)}_\Lambda$ have a unique gapped MPS ground state. If $\Lambda$ is large enough, then $H^{(0)}_\Lambda$ and $H^{(1)}_\Lambda$ share the same quantum phase according to the definition in~\cite[Section II.C.1]{MPSph1}.
\end{corollary}
A weakness of this characterization of quantum phases lies in the underlying definition, which allows blocking of physical sites.
This sheds translational invariance and the local
structure of any particles grouped into one block. As a result this approach is not suitable for the study of phases with \lq\lq{}spontaneous translational symmetry breaking\rq\rq{}. These issues are
addressed in~\cite{MPSph2}.

In this paper we focused on perturbations $\Phi_\Lambda=\sum_x\phi_x$ that have fixed interaction range. In~\cite{Stable} more general perturbations are studied, namely it is only assumed that $\phi_x=\sum_y^{N-1}\phi_{x,y}$ with interactions $\phi_{x,y}$ that act on an interval $[x-y,x+y]$ and $\norm{\phi_{x,y}}{}\leq f(y)$ for some sufficiently fast decaying function $f$. In particular, it is shown that decay faster than $f(y)=J(1+y)^{-3}$ for suitable $J$ is sufficient for stability.
In this context it is important to note that in Theorem~\ref{lem:yar} we have $\delta=\delta(\Lambda_1)$, where $\Lambda_1$ denotes the interaction range of $\Phi$. Our proof of stability shows that if classical Hamiltonians are stable under rapidly decaying perturbations in the sense of Theorem~\ref{lem:yar} this property carries over to the parent Hamiltonian
model. To generalize our discussion one could extend the derivation of Theorem 1 to analyse the dependency $\delta(\Lambda_1)$, which however was already studied in more general context~\cite{trrrr}. Hence we can conclude that under sufficiently fast decaying perturbations stability still holds.

\begin{acknowledgements} 
We acknowledge financial support from the QCCC programme of the Elite Network of Bavaria, the CHIST-ERA/BMBF project CQC and the Alfried Krupp von Bohlen und Halbach-Stiftung.
\end{acknowledgements} 
%
% BibTeX users please use

\end{document}